\documentclass[DIV=12]{scrartcl}
\usepackage[utf8]{inputenc}
\usepackage[T1]{fontenc}
\usepackage{lmodern}

\usepackage{graphicx}
\usepackage{amssymb,amsmath,amstext,amsfonts,amsthm}
\usepackage{float}
\usepackage{multirow}
\usepackage{subcaption}
\graphicspath{ {./images/} }
\usepackage[dvipsnames,svgnames,table]{xcolor}
\usepackage[colorlinks=true,linkcolor=RoyalBlue,urlcolor=RoyalBlue,citecolor=PineGreen]{hyperref}
\usepackage[authoryear]{natbib}
\usepackage{pgf,tikz}
\usetikzlibrary{calc}
\usetikzlibrary{backgrounds}
\usetikzlibrary{shapes.misc}
\usetikzlibrary{shapes.geometric}
\usepgflibrary{arrows.meta}
\usetikzlibrary{shapes.callouts}
\usetikzlibrary {shapes.arrows}
\usetikzlibrary {intersections}
\usetikzlibrary{angles,quotes}
\usepackage{booktabs}
\usepackage[nameinlink]{cleveref}
\usepackage[fixed]{fontawesome5}
\usepackage{multirow}

\theoremstyle{plain}
\newtheorem*{theorem*}{Theorem}
\newtheorem{theorem}{Theorem}
\numberwithin{theorem}{section}
\newtheorem{proposition}[theorem]{Proposition}
\newtheorem{lemma}[theorem]{Lemma}

\newtheorem{problem}[theorem]{Problem}

\theoremstyle{definition}
\newtheorem{definition}[theorem]{Definition}
\newtheorem{remark}[theorem]{Remark}

\theoremstyle{definition}

\newcommand{\R}{\mathbb{R}}

\newcommand{\tr}{\text{trace}}
\newcommand{\rank}{\text{rank\,}}

\newcommand{\diff}{\mathrm{d}}

\newcommand{\edge}[2]{#1#2}

\newcommand{\bpm}{\mathbin{\tikz[baseline=-5pt]{\node[font=\scriptsize] at (0,0) {\faPlus};\node[font=\scriptsize] at (0,-0.2) {\faMinus};}}}

\colorlet{colbg}{white}
\colorlet{colfg}{black}
\colorlet{colgraphv}{colfg!75!white}
\colorlet{colgraphe}{colfg!55!white}
\colorlet{colG}{DarkSeaGreen}
\definecolor{colR}{HTML}{CC6677}
\definecolor{colOo}{HTML}{DDCC77}
\colorlet{colO}{colOo!95!black}
\definecolor{colB}{HTML}{6699CC}
\colorlet{colY}{Gold!90!black}

\colorlet{col1}{colG}
\colorlet{col2}{colR}
\colorlet{col3}{colB}
\colorlet{col4}{colO}
\colorlet{col6}{BurlyWood!90!black}
\colorlet{col5}{MediumPurple!90!black!80!white}

\tikzstyle{vertex}=[fill=colgraphv,circle,inner sep=0pt, minimum size=4pt]
\tikzstyle{edge}=[line width=1.5pt,colgraphe]
\tikzstyle{backbone}=[edge,colR]
\tikzstyle{penny}=[colY,fill=colY!25!white] 
\tikzstyle{pennyg}=[colY!50!colR,fill=colY!50!colR!25!white] 
\tikzstyle{disk}=[dashed,colB,fill=colB!20!white,opacity=0.25] 
\tikzstyle{nodisk}=[dashed,colR,fill=colR!20!white,opacity=0.25]

\tikzstyle{labelsty}=[font=\scriptsize]

\tikzstyle{vertex2}=[fill=blue,circle,inner sep=0pt, minimum size=4pt]

\tikzstyle{env}=[rounded rectangle,minimum height=0.65cm,inner sep=0.3333em,align=left,minimum width=1cm,rounded corners=0pt,font=\small]
\tikzstyle{env1}=[env,rounded rectangle right arc=none,fill=colfg!50!colbg]
\tikzstyle{env2}=[env,rounded rectangle left arc=none]
\newcommand{\env}[3]{\tikz{\node[env1] at (0,0) {\color{colbg}#1};}\tikz{\node[env2,fill=#3] at (0,0) {#2};}}
\newcommand{\envh}[3]{\tikz{\node[env1,fill=#3,draw=#3] at (0,0) {\color{colbg}#1};}\tikz{\node[env2,draw=#3] at (0,0) {#2};}}
\newcommand{\easy}[2]{\env{#1}{{\scriptsize\faCheckCircle} #2}{colG}}
\newcommand{\difficult}[2]{\env{#1}{{\scriptsize\faTimesCircle} #2}{colR}}

\tikzstyle{boxplot}=[draw=colB,fill=colB!50!colbg,rounded corners=0.5pt]
\tikzstyle{medianplot}=[draw=colB,fill=colB!50!colbg,line width=1.5pt]
\tikzstyle{medianline}=[draw=colB,fill=colB!50!colbg,line width=1pt,dashed]
\tikzstyle{mediannode}=[circle,draw=colB,fill=colB!50!colbg,inner sep=0pt, minimum size =3pt]
\tikzstyle{mediannode1}=[mediannode,circle]
\tikzstyle{mediannode2}=[mediannode,rectangle]
\tikzstyle{mediannode3}=[mediannode,regular polygon, regular polygon sides=3,minimum size=4pt]
\tikzstyle{mediannode4}=[mediannode,regular polygon, regular polygon sides=5,minimum size=4pt]
\tikzstyle{mediannode5}=[mediannode,regular polygon, regular polygon sides=6,minimum size=3.75pt]
\tikzstyle{mediannode6}=[mediannode,diamond,minimum size=4pt]
\tikzstyle{timeline}=[draw=colB,fill=colB!50!colbg,line width=1pt]
\colorlet{cola}{colfg!75!colbg}
\tikzstyle{axes}=[draw=cola,-{Latex[round,width=3pt]}]
\tikzstyle{aline}=[draw=cola]
\tikzstyle{bline}=[draw=cola!10!colbg]
\tikzstyle{alabelsty}=[cola,font=\tiny]
\tikzstyle{hlabelsty}=[cola,font=\small]

\title{Single-cell 3D genome reconstruction in the haploid setting using rigidity theory}
\author{Sean Dewar \and Georg Grasegger \and Kaie Kubjas \and Fatemeh Mohammadi \and Anthony Nixon}
\date{}

\begin{document}

\maketitle

\begin{abstract}
    This article considers the problem of 3-dimensional genome reconstruction for single-cell data, and the uniqueness of such reconstructions in the setting of haploid organisms. We consider multiple graph models as representations of this problem, and use techniques from graph rigidity theory to determine identifiability.
    Biologically, our models come from Hi-C data, microscopy data, and combinations thereof. Mathematically, we use unit ball and sphere packing models, as well as models consisting of distance and inequality constraints. In each setting, we describe and/or derive new results on realisability and uniqueness.
    We then propose a 3D reconstruction method based on semidefinite programming and apply it to synthetic and real data sets using our models.
\end{abstract}

\section{Introduction}

The 3-dimensional (3D) structure of the genome plays an important role in gene regulation~\citep{dekker2008gene,uhler2017regulation} and genome misfolding is linked to disease~\citep{norton2017crossed}. Two different approaches for inferring the 3D genome structure are based on chromosome conformation capture techniques and microscopy-based techniques. The primary microscopy-based technique has been the fluorescent in situ hybridisation (FISH)~\citep{amann1990fluorescent}. However, this is limited in resolution and restricted to a region of the genome. Recently, in situ genome sequencing (IGS) \citep{payne2021situ}, which combines sequencing with imaging techniques, has enabled genome-wide high-resolution 3D genome reconstructions. Currently the availability of IGS data is limited compared to data from chromosome conformation capture experiments (such as 3C, 4C, 5C, Hi-C, ChiA-PET, etc), which record interactions between different fragments of the genome. One of the most popular chromosome conformation capture techniques is Hi-C~\citep{lieberman2009comprehensive} that records interactions between fragments of a genome on a genome-wide scale. The output of a Hi-C experiment is a Hi-C or a contact matrix, where rows and columns correspond to fragments of the genome and the entries of the matrix record the number of interactions between the fragments. The resolution of Hi-C data is the number of base pairs in a genome fragment. 

A recent review of~\citet{oluwadare2019overview} summarises over 30 different approaches for constructing the 3D genome from contact matrices. The approaches can be divided on the one hand into distance-based, count-based, and probability-based depending on how the interaction frequencies are modelled, and on the other hand into consensus, ensemble and population methods based on the structure of the output~\citep[Figure 2]{oluwadare2019overview}. Distance-based methods first turn contact counts into distances and then infer the positions of loci from the pairwise distances~\citep{zhang2013inference,belyaeva2021identifying}. Different distance-based methods vary in how the interaction frequencies are converted into distances and how 3D positions of loci are inferred from the distances. Contact-based methods directly infer the 3D structure without first turning contacts into distances~\citep{paulsen2017chrom3d,abbas2019integrating}. Probability-based methods model interaction frequencies as random variables, and apply maximum likelihood or Bayesian inference based methods for 3D genome reconstruction~\citep{hu2013bayesian,varoquaux2014statistical}. 

So far, the main focus of 3D genome reconstruction has been on population data, where the contact matrix counts the contacts in a collection of cells. Since the 3D structure can differ in different cells, one may use population data to either infer a single structure that represents the ``average'' of structures or an ensemble of structures that are consistent with the data. The introduction of single-cell Hi-C by~\citet{nagano2013single} has enabled the study of 3D genome structure at the single-cell level. Further single cell Hi-C methods were introduced for haploid organisms by \citet{ramani2017massively,stevens20173d} and for diploid organisms by \citet{tan2018three}. Minimization of polymer models for 3D genome reconstruction is used by \citet{nagano2013single,stevens20173d,wettermann2020minimal,shi2021hi,kos2021perspectives}, Bayesian inference by \citet{rosenthal2019bayesian}, and manifold optimization by \citet{paulsen2015manifold}. ChromSDE~\citep{zhang2013inference} is a semidefinite optimization based method that was developed for population data, but is also applicable to single cell data. ShRec3D~\citep{lesne20143d} combines shortest path computations with multidimensional scaling for 3D genome reconstruction and this works both for single cell and population Hi-C data.

\textbf{Our contribution.} In this paper, we study 3D genome reconstruction from single-cell Hi-C data for haploid organisms. Our main contributions to the abundant literature on this topic are establishing a connection between the 3D genome reconstruction and the mathematics of rigidity theory and the study of uniqueness of 3D constructions.

Rigidity theory is a field of mathematics that investigates whether a point configuration is uniquely determined, up to rigid transformations, from some pairwise distances between the points. One can study the rigidity of a point configuration under various different assumptions. We  propose mathematical models associated to penny graphs, unit ball graphs, or some modifications of them (these graphs are defined formally in the sections that follow) that correspond to different biological methods of measuring single-cell data. 
In each of these mathematical models, we survey known results and derive new results on realisability and uniqueness. Finally, we apply semidefinite programming to obtain 3D genome reconstruction algorithms for three of the proposed models to compare them and analyse whether the reconstruction results are consistent with the uniqueness results.

It is, perhaps, more desirable to investigate 3D genome reconstruction for single-cell data in the diploid setting. However, the reconstruction of the 3D genome from contact matrices in the diploid setting for population data poses numerous challenges~\citep{segal2022can}. To the best of our knowledge, the uniqueness of reconstructions for single-cell data has not been studied even in the haploid setting. In this sense, our work serves as a bridge towards establishing a similar theory in the diploid setting.

\textbf{Structure of the paper.} 
In \Cref{sec:preliminaries}, we establish our notation and provide the necessary background from graph rigidity theory, along with a concise overview of the problem of 3D genome reconstruction for single-cell data in the haploid setting. The primary objectives of this section are to establish our notation and offer a brief motivation for investigating the rigidity of the proposed families of graphs from a biological perspective.

\Cref{sec:unit_ball,sec:inequalities_and_equalities,sec:Penny,sec:generic_penny} present our models. Each subsequent section introduces additional constraints to the previous models and compares the new model with those presented in the preceding sections. 
In \Cref{sec:unit_ball} we use unit ball graphs to represent the threshold model, discuss realisability in this model and show that uniqueness cannot be achieved without additional constraints being imposed. \Cref{sec:inequalities_and_equalities} then adds microscopy constraints. The main contributions of this section include a proof that the realisability question can be solved directly by semidefinite programming if we allow the realisation to live in an arbitrary dimension
and a proof that inequalities alone do not determine whether there is a finite number of reconstructions.
We also consider the case of equality constraints only, that is, where the data comes only from microscopy experiments. This situation turns out to be equivalent to the well-studied problem of global rigidity for bar-joint frameworks.
In order to be able to mathematically deduce unique reconstructions, we then consider more specialised models in \Cref{sec:Penny,sec:generic_penny}. In particular, in \Cref{sec:generic_penny}, we prove theoretical results, especially about uniqueness and the structure of the possible graphs that could arise from the biological data corresponding to this model.

In \Cref{sec:Algorithm},
we adapt 3D reconstruction algorithms based on semidefinite programming previously used by \citet{zhang2013inference,belyaeva2021identifying} to three of the models considered in this paper.  We apply this algorithm to synthetic and real data sets, and compare our results with ShRec3D~\citep{lesne20143d}.

\section{Preliminaries}\label{sec:preliminaries}

We first review the basic theory of rigidity for structures comprising of universal joints linked by stiff bars. Then we define a number of biological models and point to the subsequent sections, where we analyse these models using rigidity theory.

\subsection{An introduction to graph rigidity}
\label{sec:rigidity}

One of the primary contributions of this paper lies in its application of concepts and tools from rigidity theory to gain insights into the 3D genome reconstruction of single cells. In this section, we provide a concise and accessible description of the topic of rigidity theory.

A \emph{(bar-joint) framework} $(G,\rho)$ is the combination of a finite, simple graph $G=(V,E)$ and a realisation $\rho:V\rightarrow \mathbb{R}^d$. The realisation assigns positions to the vertices (represented by universal joints with full rotational freedom) and hence lengths to the edges (unbendable straight line segments). 

We are interested in understanding the set of vectors $\rho'\in \mathbb{R}^{d|V|}$ such that the two frameworks $(G,\rho)$ and $(G,\rho')$ have the same edge lengths.
Define the rigidity map $f_G:\mathbb{R}^{d|V|}\rightarrow \mathbb{R}^{|E|}$ by putting $f_G(\rho)=(\dots, \|\rho(v_i)-\rho(v_j)\|^2,\dots)_{v_iv_j\in E}$. Then two frameworks $(G,\rho)$ and $(G,\rho')$ are called \emph{equivalent} if $f_G(\rho)=f_G(\rho')$. 
Instantly, one sees that this set of equivalent frameworks is always infinite since $\rho'$ may be obtainable from $\rho$ by a composition of Euclidean isometries (translations, rotations and reflections).
Any such $\rho'$ is now said to be \emph{congruent} to $\rho$. This is equivalent to specifying that all distances between pairs $\rho(v_i)$ and $\rho(v_j)$ are the same as the distances between the corresponding pairs $\rho'(v_i)$ and $\rho'(v_j)$. 

After discarding isometries, we ask whether there are still infinitely many $\rho'$ with the aforementioned property, or equivalently, whether $(G,\rho)$ is \emph{flexible}. If not, then there are finitely many frameworks $(G,\rho')$ equivalent to $(G,\rho)$ (modulo isometries), and $(G,\rho)$ is said to be \emph{rigid}.
Furthermore, $(G,\rho)$ is \emph{minimally rigid} if $(G-e,\rho)$ is flexible for all $e\in E$.
See \Cref{fig:rig} for small examples.

If the framework $(G,\rho)$ is flexible, then one of several equivalent definitions of flexibility \citep{AsimowRoth} is in terms of motions. A \emph{(continuous) motion} of $(G,\rho)$ is a family of continuous functions $x:(0,1)\times V \rightarrow \mathbb{R}^d$ such that $x(0)=\rho$ and $(G,x(t))$ is equivalent to $(G,\rho)$ for all $t\in (0,1)$. Then $(G,\rho)$ is flexible if and only if there exists a motion that does not arise from isometries of $\mathbb{R}^d$, i.e.\ if at least one pair of vertices of $G$ has a different distance between them in $(G,\rho)$ and in $(G,x(t))$ for some $t\neq 0$. As is standard in the literature, we consider an idealised model where we permit edges to pass through each other in such a continuous motion.

\begin{figure}[ht]
    \centering
    \begin{subfigure}[b]{0.3\textwidth}
      \centering
        \begin{tikzpicture}[scale=1.3]
            \node[vertex] (a) at (0,0) {};
            \node[vertex] (b) at (1,0) {};
            \node[vertex] (c) at (0.5,0.866025) {};
            \node[vertex] (d) at (1.5,0.866025) {};
            \draw[edge] (a) -- (b) -- (d) -- (c) -- (a) -- cycle;
        \end{tikzpicture}
        \caption{flexible}
    \end{subfigure}
    \begin{subfigure}[b]{0.3\textwidth}
      \centering
        \begin{tikzpicture}[scale=1.3]
            \node[vertex] (a) at (0,0) {};
            \node[vertex] (b) at (1,0) {};
            \node[vertex] (c) at (0.5,0.866025) {};
            \node[vertex] (d) at (1.5,0.866025) {};
            \draw[edge] (a) -- (b) -- (d) -- (c) -- (a) -- cycle;
            \draw[edge] (b) -- (c);
        \end{tikzpicture}
        \caption{minimally rigid}\label{fig:rig:min}
    \end{subfigure}
    \begin{subfigure}[b]{0.3\textwidth}
      \centering
        \begin{tikzpicture}[scale=1.3]
            \node[vertex] (a) at (0,0) {};
            \node[vertex] (b) at (1,0) {};
            \node[vertex] (c) at (0.5,0.866025) {};
            \node[vertex] (d) at (1.5,0.866025) {};
            \draw[edge] (a) -- (b) -- (d) -- (c) -- (a) -- cycle;
            \draw[edge] (b) -- (c);
            \draw[edge] (a) -- (d);
        \end{tikzpicture}
        \caption{rigid}
    \end{subfigure}
  \caption{Graphs with different rigidity properties.}
  \label{fig:rig}
\end{figure}
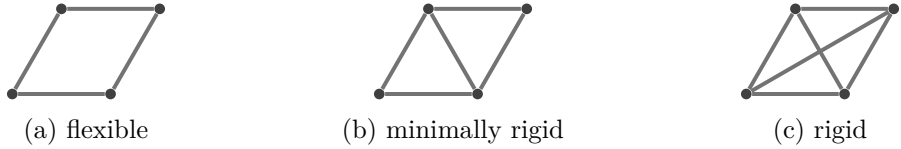

Among rigid graphs, one may ask how many different ways the structure can be realised. This is the global rigidity question, and we say that $(G,\rho)$ is \emph{globally rigid} if no other framework $(G,\rho')$ (up to Euclidean isometries) with the same edge lengths exists. In other words, $(G,\rho)$ is globally rigid if every framework in the same dimension as $(G,\rho)$ that is equivalent to it is also congruent to it. 

In each of the above definitions, it is clear that the dimension is of crucial importance. For example, every cycle $C_n$ is globally rigid (and hence rigid) on the line, but for all $n\geq 4$, the cycle $C_n$ is flexible in $\mathbb{R}^d$ for all $d\geq 2$. This motivates us to call a framework  
\emph{$d$-rigid} (resp.\ \emph{globally $d$-rigid}) if it is rigid (resp.\ globally rigid) in $\mathbb{R}^d$.

Later, we need the following elementary lemma, which has been presented for example by \citet[Lemma 2.6]{JJddim}, that relates vertex connectivity to $d$-rigidity.

\begin{lemma}\label{lem:connected}
    Let $(G,\rho)$ be $d$-rigid on at least $d+1$ vertices. Then $G$ is $d$-connected.    
\end{lemma}

A common technique in rigidity theory is to linearise by considering the Jacobian derivative $\diff f_G|_\rho$ of the rigidity map. For consistency with the literature, we put $R(G,\rho)=\frac{1}{2}\diff f_G|_{\rho}$ and refer to $R(G,\rho)$ as the \emph{rigidity matrix} of $(G,\rho)$. Then, the framework $(G,\rho)$ in $\mathbb{R}^d$ is \emph{$d$-independent} (in the literature, this is sometimes referred to as \emph{$d$-stress-free}) if the rows of $R(G,\rho)$ are linearly independent. 

\begin{lemma}[\citet{Maxwell}]\label{lem:maxwell}
    Let $(G,\rho)$ be a $d$-independent graph on at least $d$ vertices. Then $|E|\leq d|V|-\binom{d+1}{2}$ and if $|E|=d|V|-\binom{d+1}{2}$ then $G$ is $d$-rigid. 
\end{lemma}

We say that a framework $(G,\rho)$ in $\mathbb{R}^d$ is \emph{infinitesimally flexible} if there exists a function $\dot{\rho}:V\rightarrow \mathbb{R}^d$ such that
\begin{equation*}
    \langle \rho(v_i)-\rho(v_j),\dot{\rho}(v_i)-\dot{\rho}(v_j)\rangle =0 \text{ for all } v_i v_j\in E.
\end{equation*}
An infinitesimal motion $\dot{p}$ is \emph{trivial} if there exists a skew-symmetric matrix $S$ and a vector $t$ such that $\dot{p}(v_i)=Sp(v_i)+t$ for all $v_i\in V$. The framework $(G,\rho)$ is \emph{infinitesimally rigid} if every infinitesimal motion of $(G,\rho)$ is trivial. Equivalently, $(G,\rho)$ is infinitesimally rigid if $G$ is complete on at most $d+1$ vertices or $|V|\geq d+2$ and $\rank(R(G,\rho))=d|V|-\binom{d+1}{2}$.

\begin{figure}[ht]
    \centering
    \begin{tikzpicture}[scale=0.8]
        \node[vertex] (1) at (0,0) {};
        \node[vertex] (2) at (2,0) {};
        \node[vertex] (3) at (3,2) {};
        \node[vertex] (4) at (-1,2) {};
        \node[vertex] (5) at (1,2) {};
        \draw[edge] (1)--(2) (1)--(3) (1)--(4) (2)--(3) (2)--(4) (3)--(5) (4)--(5);
        \draw[thick,-{Latex[round]}] (5)++(0,0.2)--++(0,0.6);
    \end{tikzpicture}
    \caption{A graph which would be infinitesimally rigid if realised as a generic framework, but the chosen framework has a non-trivial infinitesimal motion (as indicated).}
    \label{fig:inf}
\end{figure}
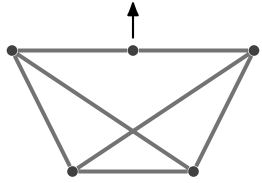

The rigidity (and global rigidity) of a framework depends on both the underlying graph, and the choice of realisation. See \Cref{fig:inf} for an example.
However, if we restrict attention to generic frameworks, then this complication disappears. A framework $(G,\rho)$ is \emph{generic} if the set of coordinates of $\rho$ forms an algebraically independent set over $\mathbb{Q}$. 

\begin{lemma}[\citet{AsimowRoth}]\label{lem:AR}
    Let $G$ be a graph on $n$ vertices.
    A generic framework $(G,\rho)$ is $d$-rigid if and only if $n\leq d$ and $G$ is complete or $n\geq d+1$ and $\rank(R(G,\rho))=d|V|-\binom{d+1}{2}$.
\end{lemma}

It follows from \Cref{lem:AR} \citep[resp.][]{GortlerHealyThurston} that if $(G,\rho)$ is a generic and $d$-rigid framework (resp.\ globally $d$-rigid framework) then every generic framework $(G,\rho')$ is rigid (resp.\ globally $d$-rigid). In other words, in the generic case, rigidity depends only on the graph. Thus, we can define $d$-rigid and globally $d$-rigid for graphs (meaning there exists a generic framework that is $d$-rigid or globally $d$-rigid).

While genericity is a strong mathematical assumption, it is worth noting that almost all realisations $\rho$ of a given graph yield generic frameworks. Therefore, we can always approximate a given realisation $\rho'$ by a generic realisation $\rho$ within a sufficiently small open neighbourhood of $\rho'$. Moreover in practical applications noisy data is likely to be generic.

In the generic case, there exist graph-theoretic characterisations of rigidity and global rigidity in the Euclidean plane, as demonstrated by \citet{Geiringer} and \citet{JacksonJordan}. These characterisations have paved the way for efficient deterministic algorithms, such as network flow algorithms, to test rigidity \citep{JHpebble}. However, extending these results to three dimensions remains a challenging open problem.  

\subsection{From contacts to distances}

In this paper, we study 3D genome reconstruction from Hi-C or contact matrices. These matrices are obtained from Hi-C experiments. The first step (cross-linking) of a Hi-C experiment is to freeze interactions between close DNA segments by chemically linking nearby chromatin segments within the same nucleus spatially. After cross-linking, the DNA segments are cut into fragments (fragmentation) using a restriction enzyme and then ligated together (ligation). This step forms the hybrid DNA molecules representing interactions between different genomic regions. Then, the DNA fragments are purified from contaminants in reverse cross-linking to facilitate DNA sequencing. Finally, using short sequences from both ends of a ligated segment, the segment is matched to a pair of interacting chromosome regions. All such interactions are recorded in a contact matrix. The rows and columns of a contact matrix correspond to regions of the genome, and the entries count the number of interactions between the corresponding regions in a Hi-C experiment.

A Hi-C matrix can record interactions either in a population of cells or in a single cell. In this paper, we focus on the single-cell setting for haploid organisms. The single-cell contact count matrix for a haploid organism is a $0/1$-matrix recording whether there is an interaction between two genomic regions. The main problem studied in this paper is the following:

\begin{problem} \label{problem:main}
    Given a single-cell Hi-C matrix and possibly some pairwise distances between genomic loci for a haploid organism, study the 3D genome reconstruction from this data. The different aspects of the 3D genome reconstruction that we study are given in~\Cref{tab:questions}.
\end{problem}

\begin{table}[ht]
    \centering
    \begin{tabular}{ p{3cm}p{0.5cm}p{7.8cm}}
         \toprule
         \multirow[t]{2}{2.5cm}{Existence} & \faDna & Given a Hi-C matrix and possibly some pairwise distances, is there a 3D genome structure with the given Hi-C matrix and distances?\\
         \cmidrule{2-3}
         & \faSquareRoot* & Given a graph and possibly some pairwise distances between its vertices, is there a realisation of the graph in a model? \\
         \midrule
         \multirow[t]{2}{2.5cm}{Uniqueness} & \faDna & Given the Hi-C matrix of a 3D genome structure and possibly some pairwise distances, is there a unique 3D genome structure with the given Hi-C matrix and distances?\\
         \cmidrule{2-3}
         & \faSquareRoot* & Given a graph in a model and possibly some pairwise distances between its vertices, is there a unique realisation of the graph in the model?\\
         \midrule
         \multirow[t]{2}{3cm}{Uniqueness for fixed edge lengths} & \faDna & Given the Hi-C matrix of a 3D genome structure and all the pairwise distances corresponding to the ones in the Hi-C matrix, is there a unique 3D genome structure?\\
         \cmidrule{2-3}
         & \faSquareRoot* & Given a graph in a model together with its edge lengths, does it have a unique realisation in the model?\\
         \midrule
         \multirow[t]{2}{2.6cm}{Reconstruction algorithm} & \faDna & Reconstruct a 3D genome structure from a Hi-C matrix and possibly some pairwise distances. \\
         \cmidrule{2-3}
         & \faSquareRoot* & Find a realisation of a graph in a model given possibly some pairwise distances between its vertices.
         \\\bottomrule
    \end{tabular}
    \caption{Biological (\faDna) and mathematical (\faSquareRoot*) questions studied in this paper.}
    \label{tab:questions}
\end{table}

We use distance-based methods to study~\Cref{problem:main}, which means that we associate distance constraints to a contact matrix. To this end, we model the genome as a string of $n$ beads. The beads correspond to genomic loci in a Hi-C experiment. The positions of the beads are recorded by a matrix $X=[x_1,\ldots,x_n]^T \in \R^{n \times d}$. We are primarily interested in the real-life setting when $d=3$.
However, we also consider other values of $d$ when obtaining complete results for $d=3$ is not feasible. Specifically, we explore less-complex models where the genome is considered to lie in a 1- or 2-dimensional space.

In the subsequent sections, we study~\Cref{problem:main} across various models that incorporate distance constraints alongside the information given by a contact matrix. These models are summarised in \Cref{tab:models}.
In the initial model, an interaction between two loci indicates that the beads corresponding to those loci are separated by a distance no greater than a threshold value $d_c$. Conversely, the absence of an interaction suggests that the beads are more than $d_c$ units apart.
This particular model, which we call the threshold model, was introduced by \citet{trieu2014large}.
While it is the most general model, we see that 3D reconstructions are not uniquely identifiable within this framework, which motivates the study of alternative models. The second assumption, that the absence of an interaction suggests that the beads are more than $d_c$ units apart, in this and other models corresponds to the idealized setting when all contacts are sampled in an Hi-C experiment. In reality, only a small set of actual contacts is sampled, so even if a contact is not observed, then it can be present. We note that if a model is not identifiable with the second assumption and has at least one solution, then it is also not identifiable without any constraints for non-contacts.

In the second model, we assume that in addition to the thresholds, we have knowledge of pairwise distances between some of the beads. This can occur, for example, in the presence of microscopy data or when estimating distances between neighbouring beads. When a sufficient number of pairwise distances are available, the 3D reconstructions become identifiable under this model. However, it may not always be possible to obtain a sufficient number of pairwise distances, which motivates the exploration of two additional models.

\begin{figure}[ht]
    \centering
    \begin{tikzpicture}[bub/.style={rectangle,rounded corners=10pt,align=left,inner sep=0.5cm,minimum width=6cm},scale=1]
        \node[bub,draw=colB,fill=colB!5!white] (h) at (0,0) {\envh{\faDna}{HiC data}{colB}\\\envh{\faSquareRoot*}{inequalities}{colB}\\[1ex] \difficult{\faMapMarker*}{unidentifiable}\\\easy{\faDatabase}{widely available}};
        \node[bub,draw=colY,fill=colY!5!white] (h) at (5,-6) {\envh{\faDna}{microscopy data}{colY}\\\envh{\faSquareRoot*}{equalities}{colY}\\[1ex] \easy{\faMapMarker*}{identifiable}\\\difficult{\faDatabase}{limited}};
        \node[bub,draw=colY!50!colB,fill=colY!50!colB!5!white] (h) at (2.5,-3) {\envh{\faDna}{combined Hi-C and microscopy data}{colY!50!colB}\\\envh{\faSquareRoot*}{inequalities and equalities}{colY!50!colB}};
        \draw[bub,draw=black] (-3.5,-10)rectangle(8.5,2.5);
        \draw[dashed] (-3.5,-8.5)--(8.5,-8.5);
        \node[anchor=west] at (-2.5,-9) {\faDna\ \ldots\ biology};
        \node[anchor=west] at (-2.5,-9.5) {\faSquareRoot* \ldots\ maths};
        \node[anchor=west] at (1.5,-9) {\faMapMarker* \ldots\ position computation};
        \node[anchor=west] at (1.5,-9.5) {\faDatabase\ \ldots\ data acquisition};
    \end{tikzpicture}
    \caption{A schematic overview depicting the two main approaches to obtaining biological data, their mathematical counterparts and the advantages~\faCheckCircle\ and disadvantages~\faTimesCircle.}
    \label{fig:overview}
\end{figure}
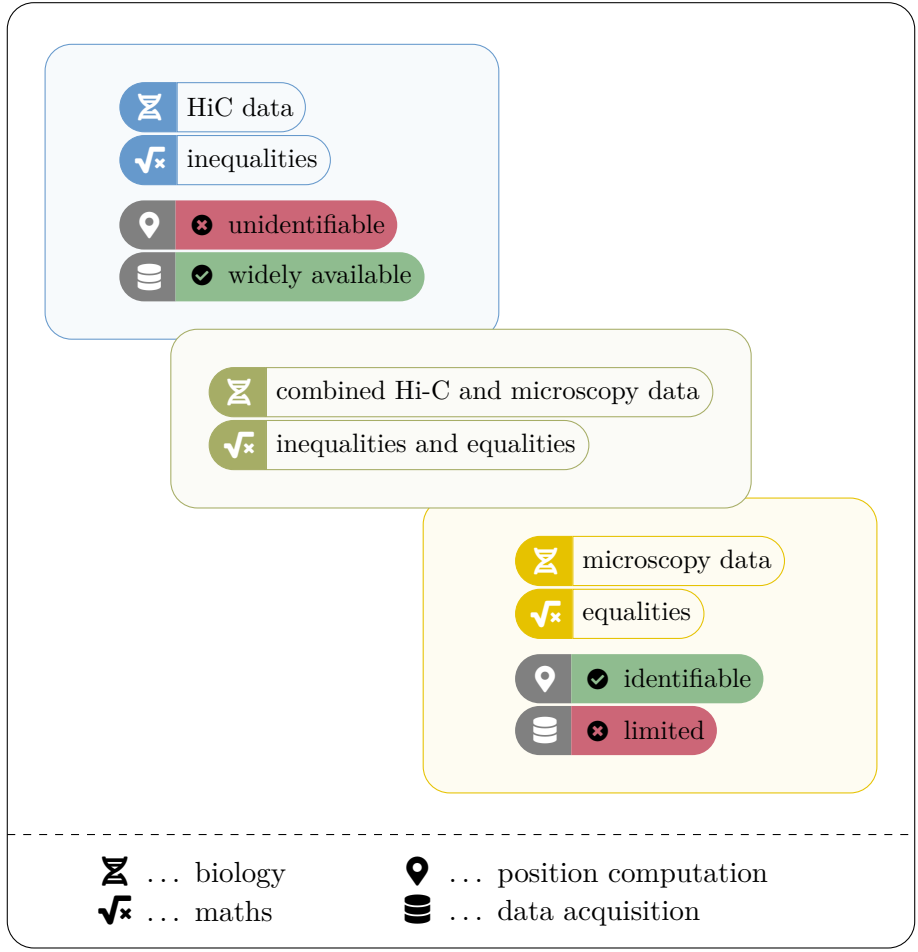

In the third model, an interaction between two loci implies that the distance between the corresponding beads precisely equals the threshold value $d_c$, as opposed to being at most $d_c$ in the threshold model. Similar to the previous models, the absence of an interaction indicates that the beads are separated by a distance greater than $d_c$. 
While this model is more restrictive than the threshold model, an advantage is that the restriction allows one to get unique 3D reconstructions in some cases.

The fourth model introduces genericity to the third one. More specifically, it assumes that an interaction between two loci implies that the distance between the corresponding beads is within an interval around the threshold $d_c$. The absence of an interaction implies a certain lower bound on the distance between the corresponding beads, which is made explicit in~\Cref{sec:generic_penny}. Additionally, both the third and fourth models can be enhanced by incorporating the extra assumption that distances between certain beads are known. Utilising all available information is consistently beneficial in all these models. 

Most of these models correspond to known frameworks in the rigidity theory. The correspondence between the biological and mathematical models is shown in~\Cref{tab:models}.
It is important to note, however, that while rigidity theory primarily focuses on edge lengths, our setting introduces additional complexity with constraints on missing edges, resulting in specific inequalities. These additional constraints make our problems more challenging and distinct from the classical rigidity theory.
\begin{table}[ht]
    \centering
    \rowcolors{2}{white}{black!5!white}
    \begin{tabular}{p{2.5cm}p{2.5cm}p{2cm}p{2cm}}
         \toprule
         biology & maths & edges & non-edges \\\midrule
         threshold & unit ball & \faLessThanEqual & \faGreaterThan\\
         threshold \& some distances & unit ball \& some distances & \faLessThanEqual \& some \faEquals & \faGreaterThan \& some \faEquals\\
         same distance & penny/marble  & \faEquals \textbf{1} & \faGreaterThan \\
         approx.\ same distance & generic radii penny & \faEquals \textbf{1}$\mathbf{\bpm\varepsilon}$ & \faGreaterThan
         \\\bottomrule 
    \end{tabular}
    \caption{The different models considered in this paper.}
    \label{tab:models}
\end{table}

\Cref{tab:models_and_questions} summarises the results on the existence and uniqueness of solutions in the different models referring to the respective theorems or literature.

\begin{table}[ht]
    \centering
    \rowcolors{2}{white}{black!5!white}
    \begin{tabular}{p{2cm}p{3.5cm}p{3.5cm}p{3.5cm}}
         \toprule
         model & existence & uniqueness & uniqueness for fixed edge lengths \\\midrule
         unit ball & \parbox[t]{3.4cm}{\Cref{lemma:existence_threshold_dim1} (dim 1)\\ \Cref{fig:forbidden-subgraphs-2d} (dim 2)\\ \Cref{comment:existence_threshold_dim3} (dim 3)\\[-2ex]} & \Cref{lemma:uniqueness_threshold_model} (no) & \parbox[t]{3.4cm}{\Cref{lemma:uniqueness_edge_lengths_threshold1}\\ \Cref{theorem:uniqueness_edge_lengths_threshold2}\\ \Cref{lemma:uniqueness_edge_lengths_threshold3}}\\
         \parbox[t]{3.4cm}{unit ball \&\\ some distances\\[-2ex]} & - & \Cref{prop:finite} & \parbox[t]{3.4cm}{\Cref{lemma:uniqueness_given_distances_equations_model}\\ \Cref{thm:ght}}\\
         penny/marble & NP-hard \citet{Breu96,Hlin97}  & \Cref{l:basic} & - \\
         interval radii penny & \Cref{lemma:existence_rigidity_interval_radii_penny} & ? & ?
         \\\bottomrule 
    \end{tabular}
    \caption{Theoretical results for the different models and questions.}
    \label{tab:models_and_questions}
\end{table}

\section{Threshold model --- unit ball graphs}\label{sec:unit_ball}

We begin by examining a model introduced by \citet{trieu2014large}, which defines a contact between two loci if their distance is below a threshold value of $d_c$. Conversely, a pair of loci is classified as non-contact if their distance exceeds the threshold. As the concepts we explore remain invariant under scalings, we can assume without loss of generality that $d_c=1$. In this particular model, a value of one in the Hi-C matrix indicates that the corresponding distance is smaller than one, while a zero represents a distance greater than one.

\begin{definition}[\citet{garamvolgyi2020global}]
    A graph $G=(V,E)$ is a \emph{($d$-dimensional) unit ball graph} if there exists a realisation $\rho:V \rightarrow \mathbb{R}^d$ such that $\|\rho(v)-\rho(w)\| \leq 1$ if and only if $vw \in E$.
    In such cases,
    the realisation $\rho$ is called a \emph{($d$-dimensional) unit ball realisation of $G$} and the pair $(G,\rho)$ is called a \emph{($d$-dimensional) unit ball framework}.
    A 1-dimensional unit ball graph is also called a \emph{unit interval graph},
    and a 2-dimensional unit ball graph is also called a \emph{unit disk graph}. See \Cref{fig:unitdisk} for an example.
\end{definition}

\begin{figure}[ht]
    \centering
    \begin{tikzpicture}
        \node[vertex] (a) at (-0.9,0) {};
        \node[vertex] (b) at (0,0) {};
        \node[vertex] (c) at (120:1) {};
        \node[vertex] (d) at (60:0.8) {};
        \draw[edge] (a)edge(b) (a)edge(c) (b)edge(c) (b)edge(d) (c)edge(d);
        \begin{scope}[on background layer]
            \foreach \v in {a,b,c,d}
            {
                \draw[disk] (\v) circle[radius=1cm];
            }
        \end{scope}
    \end{tikzpicture}
    \qquad
    \begin{tikzpicture}
        \node[vertex] (a) at (-1.2,0) {};
        \node[vertex] (b) at (0,0) {};
        \node[vertex] (c) at (120:1) {};
        \node[vertex] (d) at (60:1) {};
        \draw[edge] (a)edge(b) (a)edge(c) (b)edge(c) (b)edge(d) (c)edge(d);
        \begin{scope}[on background layer]
            \foreach \v in {a,b}
            {
                \draw[nodisk] (\v) circle[radius=1cm];
            }
        \end{scope}
    \end{tikzpicture}
    \caption{A unit disk graph with a unit disk realisation (left) and another realisation which does not satisfy the unit disk condition (right).}
    \label{fig:unitdisk}
\end{figure}
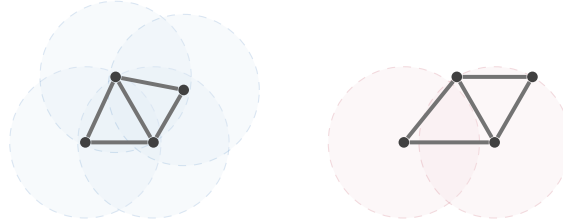

\subsection{Realisability}

Biologically the key interest is in 3D realisations however mathematically we also discuss realisations in lower dimensional spaces where there is already rich theory. For example in the literature, the realisation problem has been studied and solved for the case when $d=1$.
Specifically, a graph is a unit interval graph if and only if it is an \emph{indifference graph} \citep[for a formal definition see][]{indifference}. 

A \emph{clique} $X$ in a graph $G$ is a complete subgraph of $G$. The \emph{clique graph} of $G$ is another graph $H$ in which the vertices correspond to the maximal cliques of $G$, and two vertices are adjacent if the corresponding maximal cliques intersect in at least one vertex.

\begin{lemma}[\citet{Lekkeikerker1962}] \label{lemma:existence_threshold_dim1}
    A graph $G$ is a unit interval graph if and only if it is chordal, and every connected component of the clique graph of $G$ is a path.
\end{lemma}

While such graphs are easy to understand, when we move up to dimension $2$, the problem becomes NP-hard,  
and only partial results are available. We refer to \citet{Breu98} for a detailed analysis of the complexity and to \citet{atminas2018forbidden} for some forbidden subgraphs. In particular, the class of unit disk graphs is \emph{hereditary}, indicating closure under vertex deletions. Consequently, it raises the question of characterising unit disk graphs based on their minimal forbidden induced subgraphs. 
The known minimal forbidden subgraphs in dimension 2 are the infinite families described by \citet{atminas2018forbidden} along with the 7 graphs depicted in \Cref{fig:forbidden-subgraphs-2d}.

\begin{figure}[ht]
    \centering
    \begin{tikzpicture}[scale=0.8]
        \node[vertex] (a) at (0,0) {};
        \foreach \w in {0,60,...,300}
        {
            \node[vertex] (b\w) at (\w:1) {};
            \draw[edge] (a)edge(b\w);
        }
        \node[labelsty] at (0,-1.5) {$K_{1,6}$};
    \end{tikzpicture}
    \quad
    \begin{tikzpicture}[scale=0.8]
        \node[vertex] (a) at (-1,0) {};
        \node[vertex] (b) at (0,0) {};
        \node[vertex] (c) at (1,0) {};
        \node[vertex] (d) at (0,1) {};
        \node[vertex] (e) at (0,-1) {};
        \draw[edge] (a)edge(b) (a)edge(d) (a)edge(e) (b)edge(c) (c)edge(d) (c)edge(e);
        \node[labelsty] at (0,-1.5) {$K_{2,3}$};
    \end{tikzpicture}
    
    \begin{tikzpicture}[scale=0.8]
        \node[vertex] (a) at (-1,0) {};
        \node[vertex] (b) at (0,0) {};
        \node[vertex] (c) at (1,0) {};
        \node[vertex] (d) at (-1.5,1) {};
        \node[vertex] (e) at (-1.5,-1) {};
        \node[vertex] (f) at (0,1) {};
        \node[vertex] (g) at (0,-1) {};
        \draw[edge] (a)edge(b) (a)edge(d) (a)edge(e) (b)edge(c) (c)edge(f) (c)edge(g) (d)edge(f) (e)edge(g) (d)edge(e);
        \node[labelsty] at (-0.25,-1.5) {$G_1$};
    \end{tikzpicture}
    \quad
    \begin{tikzpicture}[scale=0.8]
        \node[vertex] (a) at (-1,0) {};
        \node[vertex] (b) at (0,0) {};
        \node[vertex] (c) at (1,0) {};
        \node[vertex] (d) at (0,0.5) {};
        \node[vertex] (e) at (-2,0) {};
        \node[vertex] (f) at (0,1) {};
        \node[vertex] (g) at (0,-1) {};
        \draw[edge] (a)edge(b) (a)edge(d) (a)edge(e) (b)edge(c) (c)edge(f) (c)edge(g) (d)edge(f) (e)edge(g) (f)edge(e);
        \node[labelsty] at (-0.25,-1.5) {$G_2$};
    \end{tikzpicture}
    \quad
    \begin{tikzpicture}[scale=0.8]
        \node[vertex] (a) at (-1,0) {};
        \node[vertex] (b) at (0,0) {};
        \node[vertex] (d) at (-0.5,1) {};
        \node[vertex] (e) at (-0.5,-1) {};
        \node[vertex] (f) at (0.5,1) {};
        \node[vertex] (g) at (0.5,-1) {};
        \draw[edge] (a)edge(d) (a)edge(e) (b)edge(d) (b)edge(e) (d)edge(f) (e)edge(g) (f)edge(g);
        \node[labelsty] at (-0.25,-1.5) {$G_3$};
    \end{tikzpicture}
    \quad
    \begin{tikzpicture}[scale=0.8]
        \node[vertex] (a) at (-1,0) {};
        \node[vertex] (b) at (0,0) {};
        \node[vertex] (c) at (1,0) {};
        \node[vertex] (d) at (-0.5,1) {};
        \node[vertex] (e) at (-0.5,-1) {};
        \node[vertex] (f) at (0.5,1) {};
        \node[vertex] (g) at (0.5,-1) {};
        \draw[edge] (a)edge(d) (a)edge(e) (b)edge(d) (b)edge(e) (c)edge(f) (c)edge(g) (d)edge(f) (e)edge(g);
        \node[labelsty] at (-0.25,-1.5) {$G_4$};
    \end{tikzpicture}
    \quad
    \begin{tikzpicture}[scale=0.8]
        \node[vertex] (a) at (0,0) {};
        \node[vertex] (b) at (1,0) {};
        \node[vertex] (c) at (2,0) {};
        \node[vertex] (d) at (0,1) {};
        \node[vertex] (e) at (1,1) {};
        \node[vertex] (f) at (2,1) {};
        \node[vertex] (g) at (1,2) {};
        \draw[edge] (a)edge(b) (a)edge(d) (b)edge(c) (b)edge(e) (c)edge(f) (d)edge(e) (d)edge(g) (e)edge(f) (f)edge(g);
        \node[labelsty] at (1,-0.5) {$G_5$};
    \end{tikzpicture}
    \caption{Some known minimal forbidden subgraphs in dimension two. The existence of such a subgraph prevents a graph from having a unit-disc realisation.}
    \label{fig:forbidden-subgraphs-2d}
\end{figure}
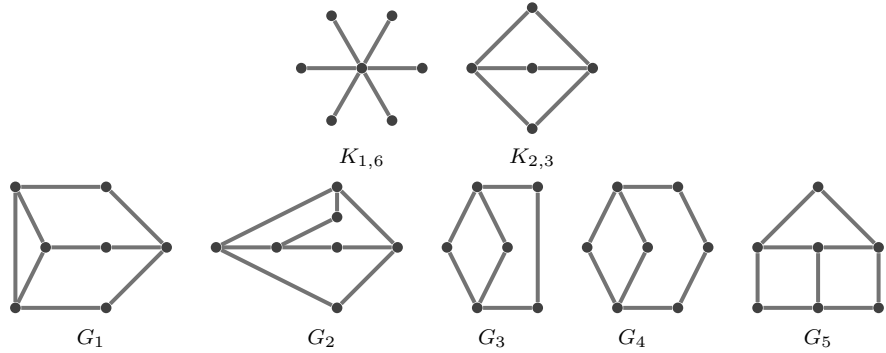

For $d\geq 3$, the problem of determining whether a graph is a unit ball graph remains NP-hard \citep{Breu98}. However, to the best of our knowledge, there is essentially no published literature
regarding forbidden subgraphs in this context. 

\begin{lemma} \label{comment:existence_threshold_dim3}
    The complete bipartite graph $K_{1,13}$ (with parts of size 1 and 13) is a minimal forbidden subgraph of the family of unit ball graphs in dimension 3.
\end{lemma}

\begin{proof}
    The kissing number in dimension three is 12 \citep[see, e.g.][]{Kissing3}. In fact, the spheres can be suitably arranged (e.g.\ in form of an icosahedron) such that they only touch the central one \citep{PfenderZiegler},
    as seen in \Cref{fig:kissingnumber}.\footnote{Due to this, it was not clear for a long time that the kissing number is indeed 12. This was a famous source of disagreement between Isaac Newton and David Gregory.
    While Newton believed that the kissing number in 3D is 12,
    Gregory (incorrectly) believed that the number was in fact 13.}
    Regardless, this shows that $K_{1,12}$ is a unit ball graph but $K_{1,13}$ is not.
    As the only proper connected subgraphs of $K_{1,13}$ with any edges are isomorphic to $K_{1,12}$,
    it follows that $K_{1,13}$ is a minimal forbidden subgraph.
    \begin{figure}[ht]
        \centering
        \includegraphics[width=5cm]{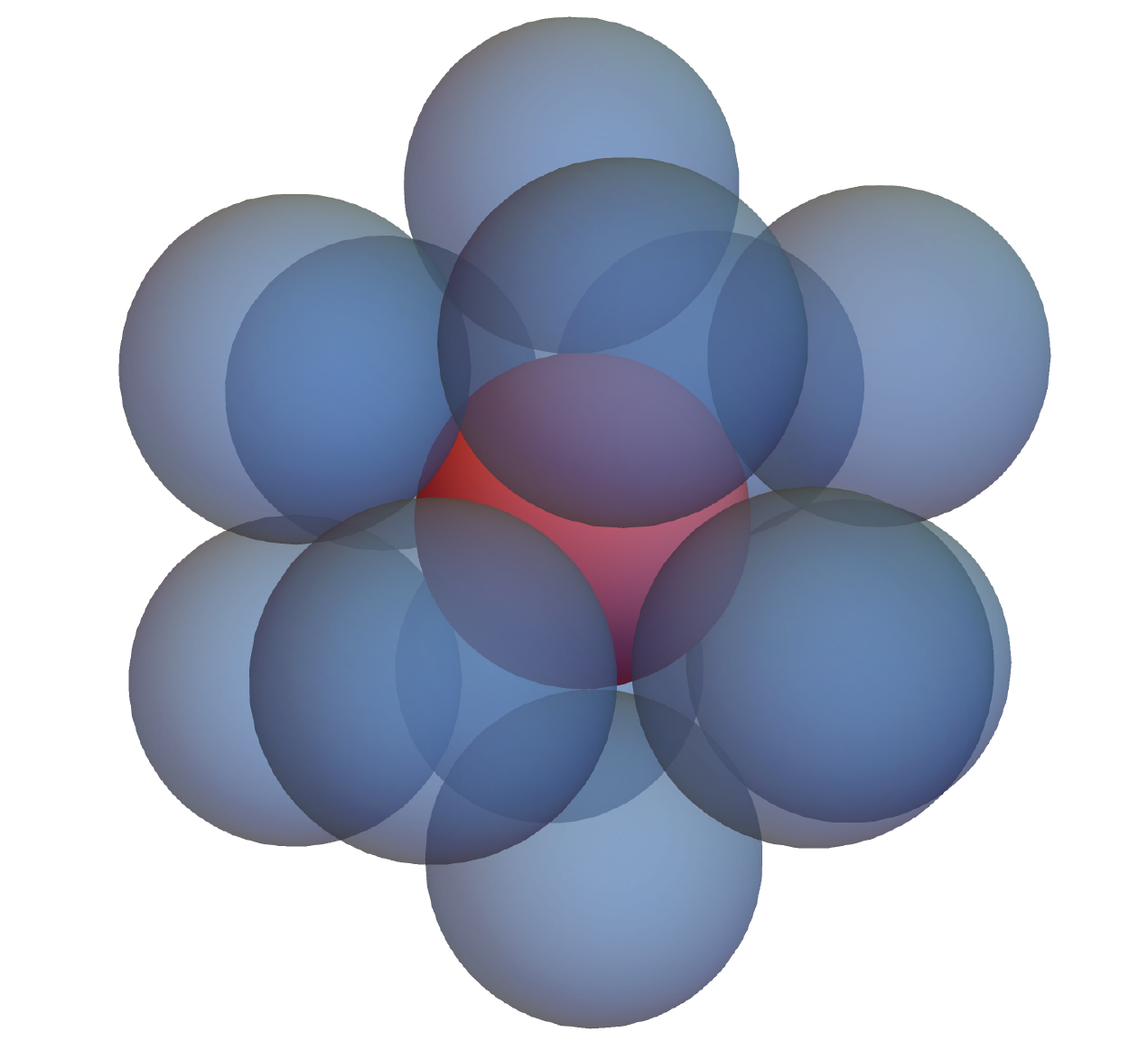}
        \caption{A central sphere with 12 spheres touching it but none of the outer spheres touch each other.}
        \label{fig:kissingnumber}
    \end{figure}
\end{proof}

We expect that there are a large number of further minimal forbidden subgraphs, some that arise from extending the examples in \Cref{fig:forbidden-subgraphs-2d}, as well as various infinite families (for example extending those given by \citet{atminas2018forbidden}). 

\subsection{Uniqueness}

In this subsection, we explore the uniqueness of reconstructions derived from the data within the Hi-C matrix. The results here apply in any dimension. Within this model, the Hi-C matrix corresponds precisely to the well-known adjacency matrix, which in turn uniquely determines the underlying graph. Consequently, there exists prior theoretical knowledge, particularly when $d=1$. However, before delving into the details, we first establish a negative result.

\begin{lemma} \label{lemma:uniqueness_threshold_model}
    Given a $d$-dimensional unit ball graph $G$, there are infinitely many $d$-dimensional unit ball realisations of $G$ modulo isometries and scalings.
\end{lemma}

\begin{proof}
    Let $G$ be a $d$-dimensional unit ball graph and let $\rho$ be a realisation of $G$. If none of the edges of $\rho$ have length equal to one, then any small enough perturbation of the realisation $\rho$ is also a unit $d$-ball realisation of $G$. If some of the edge lengths of $\rho$ are equal to one, then we first scale the realisation by $(1-\varepsilon)$. For a small enough $\varepsilon >0$, this scaling does not affect the number of $d$-dimensional unit ball realisations. Following this scaling, we may perturb as above.
\end{proof}

The above lemma indicates that we cannot obtain uniqueness of 3D reconstructions unless we pose further constraints. In the following sections, we explore what kind of additional constraints can be added to obtain uniqueness of 3D reconstructions. 
Before doing that, we briefly describe results from \citet{garamvolgyi2020global}, where they fix the edge lengths and from that are able to determine rigidity and global rigidity for some families of unit ball graphs.

A $d$-dimensional unit ball framework $(G,\rho)$ is \emph{unit ball globally rigid} in $\mathbb{R}^d$ if whenever $(G,\rho')$ is an equivalent unit ball realisation of $G$ in $\mathbb{R}^d$, then $(G,\rho')$ is congruent to $(G,\rho)$.

\begin{lemma}[{\citet[Lemma 3.1]{garamvolgyi2020global}}]\label{lemma:uniqueness_edge_lengths_threshold1}
    A generic unit ball globally rigid framework $(G,\rho)$ in $\mathbb{R}^d$ is rigid.    
\end{lemma}

A minimally $d$-rigid graph is \emph{special} if every proper subgraph that is $d$-rigid is a complete graph. Using this definition the following result gives infinite families of unit ball globally rigid graphs when $d=2$ and when $d=3$.

\begin{theorem}[{\citet[Lemma 5.1, Theorem 4.6 and Subsection 7.3]{garamvolgyi2020global}}] \label{theorem:uniqueness_edge_lengths_threshold2}
    Let $G$ be a graph.
    \begin{enumerate}
        \item If $G$ is a special minimally 2-rigid graph,
        then $G$ has a generic realisation in $\mathbb{R}^2$ that is unit ball global.
        \item If $G$ is a 4-connected maximal planar graph that is also a unit ball graph in $\mathbb{R}^3$,
        then $G$ has a generic realisation in $\mathbb{R}^3$ as a globally rigid unit ball graph.
    \end{enumerate}
\end{theorem}

\section{Threshold and microscopy - inequalities and some equalities}\label{sec:inequalities_and_equalities}

Our next model considers both inequality and equality constraints. We first discuss realisability and then discuss the usefulness of inequality constraints for reducing the number of realisations. The section concludes with a brief analysis of uniqueness in the special case when all the constraints are equalities.

\subsection{Realisability}

Let us, momentarily, ignore the issue of dimensionality for the solutions to our constraint systems.
This allows us to use semidefinite programming to solve the realisability question.

\begin{proposition}\label{p:realisability}
    Let $G=(V,E)$ be a graph,
    let $\lambda :E \rightarrow \mathbb{R}_{\geq 0}$,
    and let $s : E \rightarrow \{ -1,0,+1 \}$.
    Then determining if there exists a realisation $\rho$ of $G$ in $\mathbb{R}^{|V|-1}$ such that
    \begin{equation*}
        \|\rho(v) - \rho(w) \|^2 \,
        \begin{cases}
            \leq \lambda(vw) &\text{if } s(vw) = -1,\\
            = \lambda(vw) &\text{if } s(vw) = 0,\\
            \geq \lambda(vw) &\text{if } s(vw) = +1
        \end{cases}
    \end{equation*}
    can be achieved using a semidefinite program.
\end{proposition}

\begin{proof}
    Label the vertices of $G$ as $\{1,\ldots,n\}$.
    Let $\mathcal{M}$ be the closed convex set of all $n \times n$ symmetric matrices $M$ where $M_{ii}=0$ for all $1 \leq i \leq n$,
    and 
    \begin{equation*}
        M_{ij} \,
        \begin{cases}
            \leq \lambda(ij) &\text{if } s(ij) = -1,\\
            = \lambda(ij) &\text{if } s(ij) = 0,\\
            \geq \lambda(ij) &\text{if } s(ij) = +1
        \end{cases}
    \end{equation*}
    for each edge $ij \in E$.
    Let $\pi$ be the linear map from the space of $n \times n$ symmetric matrices to the space of $(n-1) \times (n-1)$ symmetric matrices where $\pi(M)_{ij} = M_{in} + M_{jn} - M_{ij}$ for each $1 \leq i,j \leq n-1$.
    The set $\pi(\mathcal{M})$ is now a closed convex set.
    By \citet[Theorem 1]{schoenberg1935remarks},
    there exists a realisation $\rho$ of $G$ in $\mathbb{R}^{|V|-1}$ satisfying our desired inequalities if and only if there exists a positive semidefinite matrix $A \in \pi(\mathcal{M})$.
    Hence, determining the existence of a suitable $\rho$ is equivalent to finding a feasible solution to a semidefinite program.    
\end{proof}

The two main issues with \Cref{p:realisability} are:
(i) it does not allow for strict inequalities,
and (ii) it does not guarantee the dimension of the realisation which is found.
The first issue can be (partially) resolved by replacing all strict inequalities with non-strict inequalities with additional epsilon terms.
The second issue is significantly more difficult to deal with.
The problem lies in that any equality/inequality for the dimension corresponds to a non-convex polynomial objective function for the optimisation problem.
One method for forcing the dimension of realisations to be contained within the desired dimension (e.g., 3-dimensional space) is to solve the semidefinite program given in \Cref{p:realisability}, and then `project' the realisation into the correct dimension. Although this solves the issue of dimension, it can create issues with the required edge-length bounds for the realisation. A follow-up local optimisation can then be performed using gradient descent to improve the output and (if the bounding constraints are not met) obtain a better approximation of a true solution. We analyse one variant of this method in more detail in \Cref{subsec:optform}.

\subsection{Uniqueness}

Due to the non-uniqueness of 3D reconstructions under the unit ball model, we explore the additional constraints that lead to a finite number or unique reconstructions. The unit ball graph model is characterised by a specific set of distance inequalities for edges. 
In this section, we extend our models, by allowing more general distance inequalities and incorporating equalities defined by distance constraints.
 
One case in which distance constraints arise is when we assume that the edge lengths between neighbouring beads are known, as depicted in \Cref{fig:backboneconnection}. We call the graph consisting of all beads on a chromosome and edges between neighbouring beads a \emph{backbone}. These edge lengths can be estimated based on the length of a chromosome and the resolution of bins that correspond to individual beads. This assumption was made by \citet{belyaeva2021identifying}.

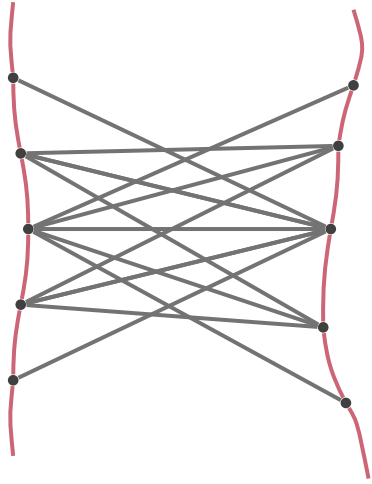
\begin{figure}[ht]
    \centering
    \begin{tikzpicture}
        \node[vertex] (a3) at (0,0) {};
        \node[vertex] (a4) at (-0.10,1) {};
        \node[vertex] (a5) at (-0.2,2) {};
        \node[vertex] (a2) at (-0.1,-1) {};
        \node[vertex] (a1) at (-0.2,-2) {};
        
        \begin{scope}[xshift=4cm]
            \node[vertex] (b3) at (0,0) {};
            \node[vertex] (b4) at (0.1,1.1) {};
            \node[vertex] (b5) at (0.3,1.9) {};
            \node[vertex] (b2) at (-0.1,-1.3) {};
            \node[vertex] (b1) at (0.2,-2.3) {};
        \end{scope}
        
        \draw[backbone] (a1)to[bend left=5](a2) (a2)to[bend left=-5](a3) (a3)to[bend left=-5](a4) (a4)to[bend left=5](a5);
         \draw[backbone] (b1)to[bend left=10](b2) (b2)to[bend left=5](b3) (b3)to[bend left=-5](b4) (b4)to[bend left=5](b5);
        
        \foreach \x in {1,2,3,4,5}
        {
            \draw[edge] (a3)edge(b\x);
        }
        \foreach \x in {1,2,4,5}
        {
            \draw[edge] (b3)edge(a\x);
        }
        \foreach \x in {2,3,4}
        {
            \draw[edge] (a2)edge(b\x);
            \draw[edge] (a4)edge(b\x);
        }
        
        \draw[backbone] (b5) .. controls ($(b5)+(0.15,0.5)$) .. ($(b5)+(0,1)$);
        \draw[backbone] (b1) .. controls ($(b1)+(0.15,-0.25)$) .. ($(b1)+(0.3,-1)$);
        \draw[backbone] (a5) .. controls ($(a5)+(-0.05,0.5)$) .. ($(a5)+(0,1)$);
        \draw[backbone] (a1) .. controls ($(a1)+(-0.05,-0.5)$) .. ($(a1)+(0,-1)$);
    \end{tikzpicture}
    \caption{Two backbones with additional edges corresponding to contacts between pairs of loci. Lengths of the red edges are assumed to be known while grey edges are only known to be smaller than a threshold.}
    \label{fig:backboneconnection}
\end{figure}

Another situation that provides edge lengths arises when in addition to Hi-C data, we also have microscopy data, e.g.\ from Fluorescence In Situ Hybridisation (FISH)~\citep{amann1990fluorescent} experiments. While FISH is limited to studying a specific region of genome, these two types of data (Hi-C and microscopy data) complement each other and, in certain cases, can yield unique reconstructions.  
This approach was proposed by \citet{abbas2019integrating}.

\begin{proposition}\label{prop:finite}
    Let $X =\{x_1,\ldots,x_n\}$ be a generic set of points in $\mathbb{R}^d$ that satisfies a set of distance constraints and distance inequalities between pairs of points.
    Further, suppose that no distance inequality is satisfied as an equality,
    e.g.\ if our distance inequality states that $\|x_i-x_j\| \leq r$ then $\|x_i-x_j\|<r$.
    Then there exist finitely many sets of points satisfying the same set of distance constraints and strict distance inequalities (modulo isometries) if and only if there exist finitely many sets of points satisfying the same set of distance constraints (modulo isometries).
\end{proposition}

\begin{remark}
    Recall from \Cref{sec:rigidity} that, if there exist only finitely many sets of points satisfying the same set of distance constraints (modulo isometries), then we say that the corresponding bar-joint framework is rigid. Further, a unique solution means that the bar-joint framework is globally rigid. 
\end{remark}

\begin{proof}[{\textbf{Proof of \Cref{prop:finite}}}]
    If there exist finitely many sets of points satisfying the same set of distance constraints (modulo isometries), then there are also finitely many sets of points satisfying the same set of distance constraints and strict distance inequalities (modulo isometries). For the converse direction, assume that there are infinitely many sets of points satisfying the same set of distance constraints (modulo isometries). Select one set that is also satisfying the inequality constraints and consider the corresponding generic bar-joint framework $(G,\rho)$.
    
    Let $f_G^{-1}(f_G(\rho))$ denote the configuration space, that is the set of all $q\in \mathbb{R}^{d|V|}$ such that $(G,\rho)$ and $(G,q)$ have the same edge lengths. 
    The fact that $f_G^{-1}(f_G(\rho))$ is a smooth manifold is non-trivial but it is well-known in rigidity theory. This follows from a basic lemma in differential topology \citep[e.g.][p. 11, Lemma 1]{Milnor1965} after proving that $f_G(\rho)$ is a regular value of $f_G$ for any generic $\rho$, where $f_G$ is the rigidity map. Establishing that $f_G(\rho)$ is a regular value is an application of both Sard's theorem and Tarski-Seidenberg elimination theory \citep[see, e.g.,][Lemma 11]{JMN14}. 

    Let $\sim$ be the equivalence relation given by congruences of realisations,
    and let $U \subset \mathbb{R}^{d|V|}$ be an open neighbourhood of $\rho$.
    We claim that, for sufficiently small $U$, the set $f_G^{-1}(f_G(\rho)) \cap U/\!\!\sim$ is a smooth manifold. 
    Since the Euclidean isometries form a Lie group that acts freely on the set of realisations in general position\footnote{In that for every isometry $g$ and every general position realisation $\rho$, the realisation $g \circ \rho$ is equal to $\rho$ if and only if $g$ is the identity map. Every generic realisation is in general position, and the set of general position realisations forms an open dense set.},
    $f_G^{-1}(f_G(\rho)) \cap U/\!\!\sim$ is also a smooth manifold for sufficiently small $U$.
    We now note that $f_G^{-1}(f_G(\rho))\cap U/\!\!\sim$ has positive dimension:
    either $G$ is connected and so $f_G^{-1}(f_G(\rho))\cap U/\!\!\sim$ is compact, and so must have positive dimension since it contains infinitely many points, or $G$ is disconnected,
    in which it is a simple exercise to show that $f_G^{-1}(f_G(\rho))\cap U/\!\!\sim$ has positive dimension stemming from translating exactly one connected component of the framework $(G, \rho)$. 
    
    Since the framework satisfies the inequalities strictly, all perturbations of the framework in a sufficiently small open neighbourhood of $\rho$ also satisfy the inequalities. Since $f_G^{-1}(f_G(\rho))\cap U/\!\!\sim$ has positive dimension, there are infinitely many edge length preserving perturbations modulo isometries, which completes the proof.
\end{proof}

While inequalities alone do not determine whether there is a finite number of reconstructions, they can reduce the solution space from a finite number to exactly one. An elementary example is given by \Cref{fig:rig:min} which as depicted contains only equalities which leads to exactly two solutions (modulo isometries). However, if we introduce an inequality constraining the distance between the two non-adjacent vertices, then we can make the solution unique.

Whether a set of distance constraints characterises a finite set of $|V|$ points can be checked with high probability by examining the rank of the corresponding rigidity matrix $R(G,\rho)$. If the rank is not $d|V|-\binom{d+1}{2}$ then the constraints do not locally characterise the points up to rigid transformations (\Cref{lem:AR}). On the other hand, if the rank is $d|V|-\binom{d+1}{2}$ then the rank is unchanged if we move to a generic framework $(G,\rho')$, and so \Cref{lem:AR} shows that the constraints do characterise the points.

\subsection{Equality constraints}

A special case of this setting is when we have only equality constraints, i.e.\ we have data from microscopy experiments but no Hi-C data. In this case, the uniqueness of a solution is equivalent to global rigidity (and finiteness is equivalent to rigidity) of bar-joint frameworks. Here, beyond the material we already presented in \Cref{sec:rigidity}, there is a large literature that one can exploit \citep[see][inter alia]{AsimowRoth,Con05,GortlerHealyThurston,Hen,JacksonJordan}.
In particular let us note the following extension of \Cref{lem:maxwell} to global rigidity which is a consequence of a theorem of Hendrickson.

\begin{lemma}[\citet{Hen}] \label{lemma:uniqueness_given_distances_equations_model}
    Let $G$ be a graph on at least $d+2$ vertices.
    If $G$ is globally $d$-rigid, then $|E|\geq d|V|-\binom{d+1}{2}+1$. 
\end{lemma}

There exists a random polynomial time algorithm for testing whether a graph is $d$-rigid. This algorithm \citep{JHpebble} follows from \Cref{lem:AR}. A similar situation exists for global $d$-rigidity, but we need another concept first.

An \emph{equilibrium stress} $\omega$ of a framework $(G,\rho)$ is a vector in the cokernel of $R(G,\rho)$. The \emph{stress matrix} $\Omega$ is the $|V|\times |V|$ symmetric matrix whose off-diagonal $ij$-entry is 0 if $ij$ is not an edge and $-\omega_{ij}$ if $ij$ is an edge, and the diagonal entry in row $i$ is the sum $\sum_{j}\omega_{ij}$. In other words, $\Omega$ is the Laplacian matrix of $G$ weighted by $\omega$.

\begin{theorem}[\citet{Con05}, \citet{GortlerHealyThurston}]\label{thm:ght}
    A graph $G$ on at least $d+2$ vertices is globally $d$-rigid if and only if $\rank(\Omega)=|V|-d-1$.
\end{theorem}

Efficient deterministic algorithms are only known for cases where $d\leq 2$, based on a combinatorial characterisation by \citet{JacksonJordan} that extends the pebble game algorithm introduced by \citet{JHpebble}. However, an alternative approach is presented by \citet{GortlerHealyThurston}, which provides a polynomial-time randomised algorithm in arbitrary dimensions leveraging \Cref{thm:ght}.

By taking on board information about the backbone of the chromosome, one can infer more detailed rigidity theoretic information.
To illustrate this, we consider graph powers.
The $k$-th \emph{power} of a graph $G$, denoted $G^k$, is obtained from $G$ by adding a new edge $uv$ for all non-adjacent vertex pairs $u, v$ of $G$ with distance at most $k$ in $G$.

The purpose of the next result about graph powers is to illustrate that basic tools from rigidity theory can provide precise information in cases when the locality of the backbone vertices forces enough edges.
Since it requires no extra work, we present the result in general dimension.

\begin{lemma} \label{lemma:uniqueness_edge_lengths_threshold3}
    Let $G=(V,E)$ be a path. Then $G^{d+1}$ is globally $d$-rigid. 
\end{lemma}

\begin{proof}
    If $|V|\leq d+1$ then $G^{d+1}$ is complete, and hence globally $d$-rigid. Suppose $|V|\geq d+2$ and note that 
    the complete graph $K_{d+2}$ is globally $d$-rigid. Observe that $G^{d+1}$ is obtained from $K_{d+2}$ by a sequence of degree $d+1$ vertex additions.
    Adding a vertex of degree $d+1$ preserves global rigidity as long as the $d+1$ neighbours affinely span $\mathbb{R}^d$ (which trivially holds in the generic case). So, the proof follows by an elementary induction argument. 
\end{proof}

Returning to the question of realisability, now for systems with only equality constraints, we remark that reducing from a high dimensional realisation, as in \Cref{p:realisability}, i.e.\ through rank relaxation makes the semidefinite optimisation problem non-convex and, in the worst case, NP-hard. However, this \emph{distance geometry problem} is well known and well studied with many special cases of applied interest being manageable in practice. We direct the interested reader to \citet{liberti2014euclidean, cassioli2015algorithm} and the references therein for details.

\section{Sphere packing models}\label{sec:Penny}

In this section, we consider a simplification of the unit ball graph model, where contacts are defined by pairwise distances of one, and non-contacts are characterised by distances greater than one.  Although this model imposes greater restrictions, it offers the advantage of producing unique reconstructions in certain cases. The model is defined by distance equality and inequality constraints, so it is a special case of the model in \Cref{sec:inequalities_and_equalities}.

\begin{definition}\label{def:penny}
    A \emph{penny graph}\footnote{Penny graphs are also known as unit coin graphs, minimum-distance graphs, smallest-distance graphs and closest-pairs graphs in the literature.} is the contact graph of a collection of unit discs with non-overlapping interiors, i.e.\ there exists a realisation $\rho\colon V\rightarrow \R^2$ such that
    \begin{align*}
        \edge{v}{w}\in E &\Longrightarrow ||\rho(v)-\rho(w)||=1,\\
        \edge{v}{w}\not\in E &\Longrightarrow ||\rho(v)-\rho(w)||>1;
    \end{align*}
    such a realisation $\rho$ is said to be a \emph{penny graph realisation} of $G$.
    \emph{Marble graphs} and \emph{marble graph realisations} are the 3D analogues of penny graphs and penny graph realisations, respectively.
\end{definition}
\Cref{fig:penny} illustrates an example and a non-example of penny realisations.

\begin{figure}[ht]
    \centering
    \begin{tikzpicture}[scale=1.2]
        \node[vertex] (a) at (0,0) {};
        \node[vertex] (b) at (1,0) {};
        \node[vertex,rotate around={60:(a)}] (d) at (b) {};
        \node[vertex] (c) at ($(d)+(1,0)$) {};
        \begin{scope}[on background layer]
            \foreach\n in {a,b,c,d}{\draw[penny] (\n) circle[radius=0.5cm];}
        \end{scope}
        \draw[edge] (a)edge(b) (b)edge(c) (c)edge(d) (d)edge(a) (b)edge(d);
        \node[colG] at ($(a)!0.5!(b)-(0,0.75)$) {\faCheck};
    \end{tikzpicture}
    \qquad
    \begin{tikzpicture}[scale=1.2]
        \node[vertex] (a) at (0,0) {};
        \node[vertex] (b) at (1,0) {};
        \node[vertex,rotate around={50:(a)}] (d) at (b) {};
        \begin{scope}[on background layer]
            \foreach\n in {a,b,d}{\draw[penny,opacity=0.5] (\n) circle[radius=0.5cm];}
        \end{scope}
        \draw[edge] (a)edge(b) (d)edge(a);
        \node[colR] at ($(a)!0.5!(b)-(0,0.75)$) {\faTimes};
    \end{tikzpicture}
    \caption{A penny realisation of a graph (left) and a a realisation of another graph that does not fulfil the penny condition since two vertices are too close (right).}
    \label{fig:penny}
\end{figure}
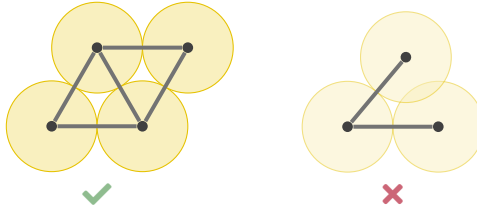

Let us first focus on penny graphs.
As penny graphs are contact graphs of disc packings,
every penny graph is planar.
The converse of this statement is not true,
as can be seen by the Moser spindle (\Cref{fig:moser}) which is not realisable as a penny graph. The next lemma lists some basic properties of penny graphs.

\begin{figure}[ht]
    \centering
    \begin{tikzpicture}[scale=1.2]
        \newcommand\w{40}
        \newcommand\ww{40}
        \node[vertex] (t) at (0,0) {};
        \node[vertex] (a1) at ($(t)+(-\w:1)$) {};
        \node[vertex] (a2) at ($(t)+(-\w-\ww:1)$) {};
        \node[vertex] (b1) at ($(t)+(180+\w:1)$) {};
        \node[vertex] (b2) at ($(t)+(180+\w+\ww:1)$) {};
        \node[vertex] (a3) at ($(a1)+(a2)-(t)$) {};
        \node[vertex] (b3) at ($(b1)+(b2)-(t)$) {};
        \draw[edge] (t)edge(a1) (t)edge(a2) (t)edge(b1) (t)edge(b2) (a1)edge(a2) (a1)edge(a3) (a2)edge(a3) (b1)edge(b2) (b1)edge(b3) (b2)edge(b3) (a3)edge(b3);
    \end{tikzpicture}
    \caption{The Moser spindle is an example of a graph that is not realisable as a penny graph.}
    \label{fig:moser}
\end{figure}
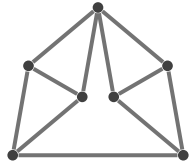

\begin{lemma}
    Let $G=(V,E)$ be a penny graph on $n$ vertices. Then:
    \begin{enumerate}
        \item\label{it:penny:bound} $|E|\leq \left\lfloor 3n-\sqrt{12n-3}\right\rfloor$;
        \item\label{it:penny:vdeg} the minimum degree is at most 3 and the maximum degree is at most 6;
        \item\label{it:penny:clique} a maximal clique has at most 3 vertices;
        \item\label{it:penny:chrom} the chromatic number is at most 4.
    \end{enumerate}
\end{lemma}

\begin{proof}
    The upper bound in \ref{it:penny:bound} was proved by \citet{Harboth}. This upper bound is also tight; it can be achieved by taking subsets of a triangular lattice. For the lower bound in \ref{it:penny:vdeg}, at least one vertex of a straight-line planar graph embedding must be convex (i.e.\ have all neighbours contained in a cone with apex at the vertex), and every convex vertex of a penny realisation has degree at most 3. The upper bound is simply the number of pennies that can be placed around a single penny.

    Since at most three points can be equidistant in the plane, conclusion \ref{it:penny:clique} is immediate.  Finally \ref{it:penny:chrom} follows from the 4-colour theorem (or a simple induction argument on the minimum degree).
\end{proof}

Significantly less is known surrounding marble graphs. We summarise the situation in the following lemma with bounds that, unlike the penny graph case, are not tight.

\begin{lemma}
    Let $G=(V,E)$ be a marble graph on $n$ vertices. Then:
    \begin{enumerate}
        \item\label{it:marble:bound} $|E|\leq 6n - 0.9 26 n^{2/3}$;
        \item\label{it:marble:vdeg} the minimum degree is at most 8 and the maximum degree is at most 12;
        \item\label{it:marble:clique} a maximal clique has at most 4 vertices;
        \item\label{it:marble:chrom} the chromatic number is at most 9.
    \end{enumerate}
\end{lemma}

\begin{proof}
    It was proven by \citet{bezdekreid} that an $n$-vertex marble graph can have at most $6n - 0.9 26 n^{2/3}$ edges.
    Since the kissing number of a sphere is 12,
    the maximum degree of a marble graph is 12.
    The minimum degree of a marble graph is at most 8;
    this follows from a similar observation that a strictly convex vertex must exist,
    and a result of Kertéz regarding the number of points on the open upper hemisphere that are at least $60^\circ$ apart \citep{Kertez}. The upper bound on the chromatic number was proved by \citet{Kertez}.
\end{proof}

The current maximum number of edges found in a marble graph is roughly $6n - \sqrt[3]{486}n^{2/3}$, which is achievable with induced subgraphs of the face-centred cubic lattice \citep{Bezdek12}.
Similar to penny graphs, we can easily identify that the clique number of a marble graph is at most 4 (since at most four points can be equidistant in $\mathbb{R}^3$).
The maximum possible chromatic number of a marble graph is between 5 \citep{Maehara} and 9 \citep{Kertez},
but the exact upper bound is currently open. 

Unfortunately,
determining whether a graph is a penny graph is NP-hard \citep{Breu96}.
Similarly,
determining whether a graph is a marble graph is NP-hard \citep{Hlin97}.
Interestingly,
the problem of identifying contact graphs of unit $d$-sphere packings (the $d$-dimensional variant of identifying penny and marble graphs) is known to be NP-hard when $d \in \{2,3,4,8,24\}$ \citep{Hlin97,HK01},
but remains completely open in all other dimensions greater than 1.

We are particularly interested in the rigidity and global rigidity properties of penny and marble graphs in relation to their penny graph and marble graph realisations. 
We illustrate this in~\Cref{fig:rigflex-penny} in examples.
This motivates the following definition.

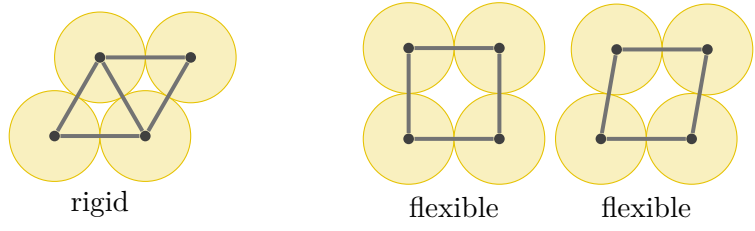
\begin{figure}[ht]
    \centering
    \begin{tikzpicture}[scale=1.2]
        \node[vertex] (a) at (0,0) {};
        \node[vertex] (b) at (1,0) {};
        \node[vertex,rotate around={60:(a)}] (d) at (b) {};
        \node[vertex] (c) at ($(d)+(1,0)$) {};
        \begin{scope}[on background layer]
            \foreach\n in {a,b,c,d}{\draw[penny] (\n) circle[radius=0.5cm];}
        \end{scope}
        \draw[edge] (a)edge(b) (b)edge(c) (c)edge(d) (d)edge(a) (b)edge(d);
        \node[] at ($(a)!0.5!(b)-(0,0.75)$) {rigid};
    \end{tikzpicture}
    \qquad\qquad
    \begin{tikzpicture}[scale=1.2]
        \node[vertex] (a) at (0,0) {};
        \node[vertex] (b) at (1,0) {};
        \node[vertex,rotate around={90:(a)}] (d) at (b) {};
        \node[vertex] (c) at ($(d)+(1,0)$) {};
        \begin{scope}[on background layer]
            \foreach\n in {a,b,c,d}{\draw[penny] (\n) circle[radius=0.5cm];}
        \end{scope}
        \draw[edge] (a)edge(b) (b)edge(c) (c)edge(d) (d)edge(a);
        \node[] at ($(a)!0.5!(b)-(0,0.75)$) {flexible};
    \end{tikzpicture}
    \begin{tikzpicture}[scale=1.2]
        \node[vertex] (a) at (0,0) {};
        \node[vertex] (b) at (1,0) {};
        \node[vertex,rotate around={80:(a)}] (d) at (b) {};
        \node[vertex] (c) at ($(d)+(1,0)$) {};
        \begin{scope}[on background layer]
            \foreach\n in {a,b,c,d}{\draw[penny] (\n) circle[radius=0.5cm];}
        \end{scope}
        \draw[edge] (a)edge(b) (b)edge(c) (c)edge(d) (d)edge(a);
        \node[] at ($(a)!0.5!(b)-(0,0.75)$) {flexible};
    \end{tikzpicture}
    \caption{A rigid penny realisation of a penny graph (left) and two flexible penny realisations of a penny graph (right). One can continuously deform the locations of the pennies, maintaining the given contacts to deform the first penny realisation into the second.}
    \label{fig:rigflex-penny}
\end{figure}

\begin{definition}
    Let $G$ be a penny graph (resp., marble graph).
    We say that $G$ is \emph{penny-rigid} (resp., \emph{marble-rigid}) if there exist finitely many penny graph realisations (resp., marble graph realisations) modulo isometries.
    We say that $G$ is \emph{globally penny-rigid} (resp., \emph{globally marble-rigid}) if there exists exactly one penny graph realisation (resp., marble graph realisation) modulo isometries.
\end{definition}

Similarly to the existence problem,
determining whether a penny graph is penny-rigid or globally penny-rigid is also computationally difficult.

\begin{lemma}[{\citet[Lemma 2.4]{penny_paper}}]\label{l:basic}
    Let $G=(V,E)$ be a sphere graph.
    Then the following properties hold.
    \begin{enumerate}
        \item\label{p:basic1} If $G$ is globally sphere-rigid, then $G$ is sphere-rigid.
        \item\label{p:basic2} If $G$ is sphere-rigid, $d \leq 3$ and $|V|\geq d+1$, then $G$ is $d$-connected.
    \end{enumerate}
\end{lemma}

\citet[Theorem 2.6]{penny_paper} give a precise, efficient characterisation of when a chordal penny graph is penny-rigid, as well as the same statement for marble-rigidity. (For chordal graphs, it is shown by \citet{penny_paper} that penny-global rigidity and penny-rigidity coincide, so the characterisation applies equally to penny-global rigidity.)
That paper goes on to illustrate many examples of how penny-rigidity and penny-global rigidity can differ from the standard rigidity notions for bar-joint frameworks \citep[see][Table 1]{penny_paper}.

\begin{lemma}\label{lem:k113nk26}
    The graphs $K_{1,13}$ and $K_{2,6}$ are minimal forbidden subgraphs of the family of marble graphs.
\end{lemma}

\begin{proof}
    A minor adaptation to the proof given in \Cref{comment:existence_threshold_dim3} proves that $K_{1,13}$ is a minimal forbidden subgraph.

    For the remainder of the proof
    we introduce the following notation.
    For any two graphs $G_1,G_2$ with distinct vertices,
    we denote by $G + H$ the graph formed by connecting every vertex in $G$ to every vertex in $H$.
    We also denote any graph with $n$ vertices by $G_n$.
    We say that a graph is a $d$-dimensional ball graph if it is the contact graph of a set of $d$-dimensional balls with non-overlapping interiors.
    \citet[Corollary 3.7]{Chen16} proved that every graph $G_4 + G_6$, with the exception of $G_4$ being a cycle and $G_6$ being the 1-skeleton of the octahedron, is not the contact graph of a set of 4-dimensional balls (here allowing for the balls to have different radii) with non-overlapping interiors.
    In particular,
    any graph of the form $K_2 + G_2 + G_6$ is not one of the aforementioned contact graphs.
    \citet[Proposition 4.5]{HK01} showed that a graph $G + K_2$ is the contact graph of a set of $(d+1)$-dimensional balls with non-overlapping interiors if and only if $G$ is the contact graph of a set of $d$-dimensional unit radius balls with non-overlapping interiors. 
    Hence,
    any graph of the form $G_2 + G_6$ is not a marble graph;
    in particular,
    $K_{2,6}$ is not a marble graph.
    Minimality now follows from noting that every vertex-induced subgraph of $K_{2,6}$ is a marble graph: $K_{1,6}$ 
    is a marble graph as it is a proper induced subgraph of the minimal forbidden subgraph $K_{1,13}$,
    and $K_{2,5}$ is a marble graph as it is the contact graph of the marble realisation pictured in \Cref{fig:k25}.
\end{proof}

\begin{figure}[ht]
    \centering
    \includegraphics[width=7cm]{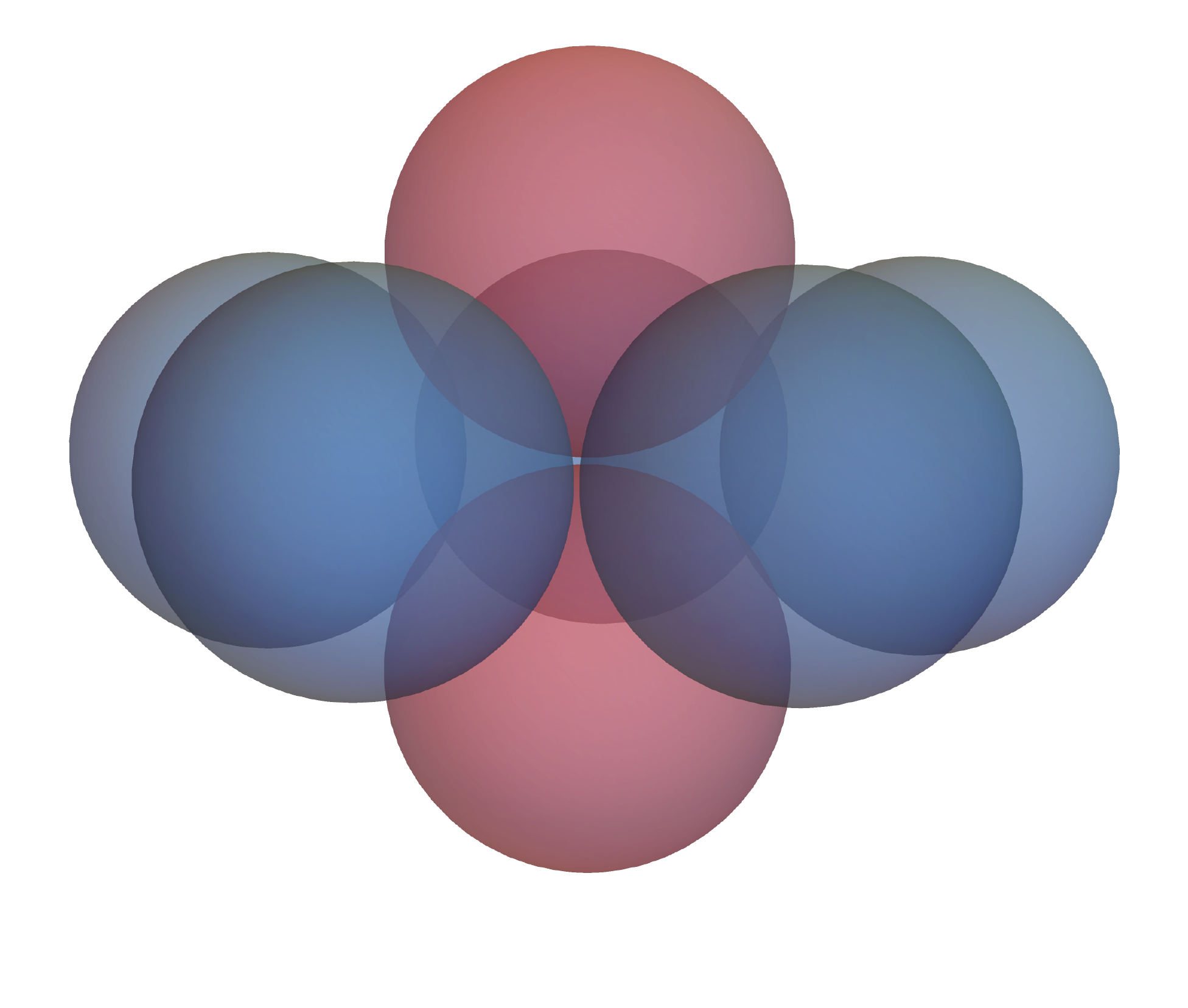}
    \caption{A marble realisation of the complete bipartite graph $K_{2,5}$.}
    \label{fig:k25}
\end{figure}

\begin{remark}
    \citet{HolmesCerfon} studied the problem of determining the number of all possible sphere-rigid marble graphs with a fixed number of vertices. Specifically, they systematically enumerated the sphere-rigid marble graphs with at most 14 vertices, and the sphere-rigid marble graphs with the maximum number of edges with at most 19 vertices. They found these graphs computationally, with no completeness proof. However, they observed that the list may only miss a negligible number of exceptionally singular graphs.
    
    Recall that \Cref{lem:AR} implies that a minimally rigid generic framework in 3-dimensions, on a graph $G$ with $n$ vertices, has $3n-6$ edges. 
    In statistical physics, the term \emph{hypostatic} is used to refer to equal radii sphere packings (or in our language, sphere-rigid marble graphs) with less than $3n-6$ edges and the term \emph{hyperstatic} is used for sphere-rigid marble graphs with more than $3n-6$ edges. The computations of \citet{HolmesCerfon} establish that there is a transition at $n=10$ where both hypostatic and hyperstatic rigid marble graphs emerge.
\end{remark}

\section{Generic radii penny and marble graphs}
\label{sec:generic_penny}

In this section, we consider a generalisation of the penny graphs from~\Cref{{def:penny}}, where the radii of the circles could be slightly larger or smaller than one.  See \Cref{fig:genericpenny} for an example of a graph that is realisable only with generic radii.
\begin{definition}
    A \emph{generic penny/marble graph} is the contact graph of generic radii circles/spheres. The vertices correspond to the centres of the circles/spheres and there is an edge between two vertices if and only if the corresponding circles/spheres touch.
\end{definition}

\begin{figure}[ht]
    \centering
    \begin{tikzpicture}
        \foreach \i [evaluate=\i as \w using 60*\i] in {1,2,3,4,5,6}
        {
            \node[vertex] (a\i) at (\w:1) {};
        }
        \node[vertex] (a0) at (0,0) {};
        \begin{scope}[on background layer]
            \foreach\n in {0,1,...,6}{\draw[penny] (a\n) circle[radius=0.5cm];}
        \end{scope}
        \foreach \i [evaluate=\i as \oi using {int(mod(\i,6)+1)}] in {1,2,3,4,5,6}
        {
            \draw[edge] (a0)edge(a\i) (a\i)edge(a\oi);
        }
    \end{tikzpicture}
    \qquad\qquad
    \begin{tikzpicture}
        \pgfmathparse{58}
        \let\mysangle=\pgfmathresult
        \pgfmathparse{0.5/(2*sin((180-\mysangle)/2)/sin(\mysangle)-1)}
        \let\mysrad=\pgfmathresult
        \pgfmathparse{0.5+\mysrad}
        \let\myrrad=\pgfmathresult
        \foreach \i [evaluate=\i as \w using \mysangle*\i] in {1,2,3,4,5,6}
        {
            \node[vertex] (a\i) at (\w:\myrrad) {};
        }
        \node[vertex] (a0) at (0,0) {};
        \begin{scope}[on background layer]
            \draw[penny] (a0) circle[radius=0.5cm];
            \foreach\n in {1,...,6}{\draw[pennyg] (a\n) circle[radius=\mysrad cm];}
        \end{scope}
        \foreach \i in {1,2,3,4,5,6}
        {
            \draw[edge] (a0)edge(a\i);
        }
        \draw[edge] (a1)edge(a2) (a2)edge(a3) (a3)edge(a4) (a4)edge(a5) (a5)edge(a6);
    \end{tikzpicture}
    \caption{A graph with a penny-realisation (left) and one that is only realisable with generic radii (right). Note that on the right the outer circles are slightly smaller than the inner one (though for the purpose of the figure still non-generic).}
    \label{fig:genericpenny}
\end{figure}
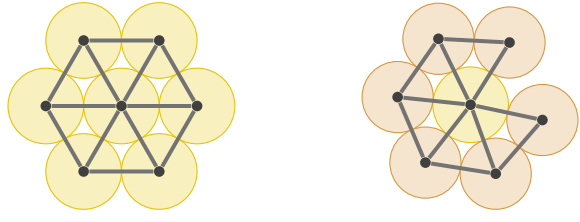

In our target application the lengths of the edges are all expected to be similar. Hence, we consider \emph{$\varepsilon$-interval radii} penny graphs, that is each penny has radius within an interval from $1-\varepsilon$ to $1+\varepsilon$ for some sufficiently small $\varepsilon>0$, and the set of all radii is algebraically independent.

\citet{stickydisks} studied the realisability and rigidity properties of generic radii penny graphs. 
In particular, they provided a tight linear bound on the number of edges in any generic radii penny graph. Further when the bound holds, they classified when the contact graph is rigid and infinitesimally rigid as a penny graph \citep[Theorem~1.5]{stickydisks}. In our setting, this translates to the following lemma. 

\begin{lemma}\label{lemma:existence_rigidity_interval_radii_penny}
    The $\varepsilon$-interval radii penny graphs have at most $2n-3$ edges for any $\varepsilon>0$ and they are rigid precisely when the bound is attained.
\end{lemma}

Observe that this upper bound is less than the minimum number of edges required for a bar-joint framework to be globally rigid in the plane (\Cref{lemma:uniqueness_given_distances_equations_model}). 

It is also worth pointing out that realisations of generic radii penny graphs do not correspond to generic frameworks. Indeed for a generic framework $(G,\rho)$ in the plane on $n$ vertices, the set of coordinates of $p$ is an algebraically independent set of size $2n$, whereas a generic radii penny graph only has an algebraically independent set of size $n$. From this viewpoint, it should be no surprise that these graphs typically exhibit non-generic behaviour. Hence, none of the classical rigidity theory methods can be directly applied in this setting.
Note that our $\epsilon$-interval radii penny graphs are even more restrictive, given that the radii of our circles are in the intervals $(1-\varepsilon, 1+\varepsilon)$.
Nevertheless, as \Cref{lem:intervalpenny} shows, global rigidity can be established in some interesting cases.

While generic radii penny graphs can have any maximum degree and adjacent vertices may have multiple common neighbours, we note the following result that illustrates the additional restrictions that interval radii can place.

\begin{lemma}
    Let $G$ be an $\varepsilon$-interval radii penny graph for some 
    \begin{equation*}
        \varepsilon \leq 1 - \sqrt{2(1 - \cos (2\pi / 7))} \approx 0.13223.
    \end{equation*}
    Then the maximum degree in $G$ is at most 6.
\end{lemma}

\begin{proof}
    Fix $\varepsilon >0$ to be the smallest value so that it is possible to have a configuration of seven pennies with disjoint interiors and radii contained in the closed interval $[1-\varepsilon,1+\varepsilon]$ such that one penny is in contact with all other pennies. We call the penny in contact with all others the \emph{centre penny} and the others the \emph{outer pennies}.
    
    We observe that it is always possible to increase the radius of the centre penny and/or decrease the radius of the outer pennies and still find a configuration of seven pennies with the same properties.
    Hence,
    we may suppose that the centre penny has radius $1+ \varepsilon$ and that each outer penny has radius $1-\varepsilon$.
    As we chose $\varepsilon$ as small as possible,
    it follows that the outer pennies must form a cycle.
    This set-up produces a regular heptagon with side lengths $2-2\varepsilon$.
   Label the centre of the centre penny as $o$. Then $o$ is distance 2 from each vertex of the heptagon.
   
    Pick two adjacent vertices $x,y$ of the heptagon.
    The triangle $xoy$ is isosceles with angle $2\pi/7$ radians at $o$ and side lengths $\|x-o\| = \|y-o\| = 2$ and $\|x-y\|=2-2\varepsilon$. 
    Using the cosine rule and rearranging, we see that $\varepsilon = 1 - \sqrt{2(1 - \cos (2\pi / 7))}$.
    The result now follows since an $\varepsilon$-interval radii penny graph does not allow pennies with radius $1 \pm \varepsilon$.
\end{proof}

\begin{lemma}\label{l:point in middle}
    Let $P_0,P_1,P_2,P_3$ be four discs with disjoint interiors, with each $P_i$ having radius $r_i$ and centre at $p_i$.
    Suppose that $P_1,P_2,P_3$ are in pairwise contact and $p_0$ lies in the interior of the triangle formed by the points $p_1,p_2,p_3$.
    Then
    \begin{equation*}
        3 + 2 \sqrt{3} ~ \leq ~ \max\{ r_i/r_0 : i \in \{1,2,3\} \}.
    \end{equation*}
\end{lemma}

\begin{proof}
    Fix $r_0 = 1$.
    Suppose that $P_1,P_2,P_3$ have been chosen so that the value $r_{\max} := \max\{ r_i : i \in \{1,2,3\} \}$ is minimal.
    Note that it is possible to shrink each of $P_1,P_2,P_3$ whilst maintaining all contacts between discs up to the point where $P_0$ is in contact with all other discs,
    and this operation decreases the value of $r_{\max}$.
    Hence, we can assume that $P_0$ is in contact with each of $P_1,P_2,P_3$.
    Suppose that $r_{\max} < 3 + 2 \sqrt{3}$.
    For any two vertices $i,j \in \{1,2,3\}$,
    let $\theta_{ij}$ be the angle at $p_0$ in the triangle $p_0p_ip_j$.
    Since $\|p_0-p_i\|= 1+ r_i$, $\|p_0-p_j\| = 1+r_j$ and $\|p_i-p_j\|=r_i+r_j$,
    it follows from elementary trigonometry that
    \begin{align*}
        \cos \theta_{ij} &= \frac{(1+r_i)^2 + (1+ r_j)^2 - (r_i+ r_j)^2}{2(1+r_i)(1+r_j)} = 1 - \frac{2r_i r_j}{(1+r_i)(1+r_j)} \\ 
        &< 1 - \frac{2(3 + 2 \sqrt{3})^2}{(4 + 2 \sqrt{3})^2} = \frac{1}{2}.
    \end{align*}
    Hence, $\theta_{ij} < 2\pi/3$.
    However, we now see that $\theta_{12} + \theta_{23} + \theta_{13} < 2 \pi$,
    contradicting that the discs $P_0,P_1,P_2,P_3$ are pairwise touching with $P_0$ at their centre.
    Hence, $r_{\max} \geq 3 + 2 \sqrt{3}$,
    which implies the desired inequality.
\end{proof}

\begin{lemma}\label{lem:0.73}
    Let $G$ be an $\varepsilon$-interval radii penny graph for some 
    \begin{equation*}
        \varepsilon \leq \sqrt{3} -1 \approx 0.73205.
    \end{equation*}
    Then any pair of adjacent vertices have at most two common neighbours.
\end{lemma}

\begin{proof}
   Let $G$ be a penny graph with vertices $x,y,a,b,c$ so that $x,y$ are adjacent to each other and each of $a,b,c$.
    Let $\rho$ be an $\varepsilon$-interval radii penny realisation of $G$,
    and fix $r_v$ to be the radii of the disc associated with vertex $v$.
    We observe now that each of the triples $xya,xyb,xyc$ is a triangle in $G$.
    By the planarity of $\rho$,
    there exists distinct $i,j \in \{a,b,c\}$ such that $\rho(i)$ is contained in the triangle formed by the points $\rho(x),\rho(y),\rho(j)$.
    By \Cref{l:point in middle} we see that
    \begin{equation}\label{eq:root3 minus 0ne}
        3 + 2 \sqrt{3} \leq \max\{ r_v/r_i : v \in \{x,y,j\} \} < \frac{1 + \varepsilon}{1-\varepsilon} \leq 3 + 2 \sqrt{3},
    \end{equation}
    which is a contradiction.    
\end{proof}

\begin{remark}
    Generic radii penny graphs can be naturally extended to higher dimensions, by replacing the disks with
    spheres. It is conjectured by \citet{ozkan2018algorithm} that the number of edges in the graph of such a generic radii packing is bounded linearly by $3n-6$. This conjecture has been further studied by \citet{DewarSphere}. Notably, it has been demonstrated that the conjecture holds true for any sphere packing with a contact graph in the form (using the notation described in \Cref{lem:k113nk26}) of $G+ K_2$, where $G$ is a penny graph containing no cycles.
\end{remark}

Generic radii global rigidity is not a generic property. That is, some graphs have some generic radii realisations that are globally rigid and other generic radii realisations that are not globally rigid. For example the realisation of $K_4-e$, depicted on the left of \Cref{fig:penny} as a penny graph, can be perturbed to a nearby generic radii realisation.
This is globally rigid as a generic radii (or interval radii) penny graph. The underlying bar-joint framework, on the other hand, is not globally rigid as one may `flip' either vertex of degree 2 through the line defined by its neighbours to obtain a second realisation with the same edge lengths. This flip would, of course, destroy the generic radii realisation as the pennies corresponding to the degree 2 vertices would overlap.

Assuming that a graph is an interval radii penny graph, the global rigidity problem now becomes: does it have a unique interval radii penny graph realisation up to rigid transformations of the plane? We give a sample theoretical result illustrating a family of graphs for which uniqueness in this model can be determined from the graph.

A \emph{2-lateration} graph is a graph obtained from $K_2$ by sequentially adding vertices of degree 2 so that, at each step, the neighbours of the new degree 2 vertex are adjacent. Such graphs are relevant to our study for example if one assumes the backbone is connected there is a reasonably high probability that there is a spanning 2-lateration subgraph.

\begin{lemma}\label{lem:intervalpenny}
    Let $(G,\rho)$ be an $\varepsilon$-interval radii penny graph realisation where $\varepsilon \leq \sqrt{3} -1 \approx 0.73205$. If $G$ is a planar 2-lateration graph then $(G,\rho)$ is globally rigid.
\end{lemma}

\begin{proof}
    The lemma trivially holds when $(G,\rho)$ is a triangle. Inductively, suppose that the lemma is true for any $\varepsilon$-interval radii penny graph realisation of a planar 2-lateration graph on up to $n$ vertices. 
    Let $(G,\rho)$ be an $\varepsilon$-interval radii penny graph realisation where $G$ is a planar 2-lateration graph with $n+1$ vertices.
    Then $G$ contains a degree 2 vertex $v$ (with neighbours $x,y$) so that $G-v$ is a planar 2-lateration graph on $n$ vertices and $(G-v,\rho|_{G-v})$ is a $\varepsilon$-interval radii penny graph realisation of $G-v$.
    By our inductive assumption,
    $(G-v,\rho|_{G-v})$ is globally rigid.
    As $G$ is a planar 2-lateration graph with 4 or more vertices,
    $x,y$ are both adjacent to another vertex $z$.
    Given the positions of the vertices of $G-v$ given by $\rho$,
    there are exactly two possible locations for the vertex $v$ that satisfy the distance constraints for the edges $vx$ and $vy$.
    As $\{x,y,z\}$ defines a triangle,
    it follows from \Cref{l:point in middle} and \cref{eq:root3 minus 0ne} that there is a unique location for the penny corresponding to $v$ so that there is no overlap. 
\end{proof}

\section{Reconstruction algorithm}\label{sec:Algorithm}

In this section, we study the 3D reconstruction problem. This problem is closely related to the Euclidean distance problem which aims to construct points from pairwise distances between them. We adapt semidefinite programming based algorithms for the Euclidean distance problem and 3D genome reconstruction to the
models in~\Cref{sec:unit_ball},~\Cref{sec:inequalities_and_equalities}, and~\Cref{sec:Penny}. It seems that our approach can not easily be applied to the model in~\Cref{sec:generic_penny}.

Numerous algorithms for solving the Euclidean distance problem are summarised in the over\-view article by \citet{liberti2014euclidean}.
In our setting, by using inequalities we allow interval data~\citep[Section 3.4]{liberti2014euclidean}, 
i.e.\ edges can be weighted with bounded or unbounded real intervals instead of just real numbers. More details on the interval setting can be found \citep[see][]{cassioli2015algorithm}.

We focus on algorithms based on semidefinite programming, as they provide convex formulations that easily allow one to account for noisy measurements. Semidefinite programs for distance geometry problems have been particularly popular in wireless sensor network localisation~\citep{biswas2006semidefinite2,so2007theory}, but also in molecular conformation problems~\citep{leung2010sdp} and computer vision~\citep{weinberger2006unsupervised}. In the context of 3D genome reconstruction, they have been used in the population haploid case by \citet{zhang2013inference} and the diploid case by \citet{belyaeva2021identifying,cifuentes20233d}.

\subsection{Optimisation formulation}\label{subsec:optform}

To construct a semidefinite relaxation for the reconstruction problem, we recall the connection between Euclidean distance and positive semidefinite matrices. Let $x_1,\ldots,x_n \in \mathbb{R}^d$ be the coordinates of beads. Let $G$ be the corresponding Gram matrix, i.e.\ $G_{ij}= \langle x_i, x_j \rangle$. It is a positive semidefinite matrix of rank at most $d$. 
The squared distance matrix $D^{(2)}$ can be obtained from the Gram matrix $G$ using the relation
\begin{equation*}
    D^{(2)}_{ij} = G_{ii}+G_{jj}-2G_{ij}.
\end{equation*}

We define
\[
    g_{ij}(G):=G_{ii}+G_{jj}-2G_{ij}.
\]
Let our model be defined by
\begin{equation} \label{eqn:constraints}
    \begin{aligned}
    &g_{ij}(G) = a_{ij}^2 \text{ for } (i,j) \in S_{=},\\
    &g_{ij}(G) \leq \overline{a}_{ij}^2 \text{ for } (i,j) \in S_{\leq},\\
    &g_{ij}(G) \geq \underline{a}_{ij}^2 \text{ for } (i,j) \in S_{\geq},\\
    &g_{ij}(G) < \overline{a}_{ij}^2 \text{ for } (i,j) \in S_{<},\\
    &g_{ij}(G) > \underline{a}_{ij}^2 \text{ for } (i,j) \in S_{>},
    \end{aligned}
\end{equation}
where $S_{=}, S_{\leq},S_{\geq},S_{<},S_{>} \subseteq [n] \times [n]$. 

For example, if we have the unit ball graph model, then $S_{=}=\emptyset$, $S_{\leq}=E$ and $S_{>} =([n] \times [n]) \backslash (\{(1,1),\ldots,(n,n)\} \cup E)$. We did not find an easy way to adapt this formulation to interval radii marble graph model, which is the reason why the model is not considered in this section.

We consider the following convex optimisation formulation:
\begin{equation} \label{formulation:existence_unit_ball_graph_optimisation_trace}
    \begin{aligned}
        & \underset{G}{\text{minimise}}
        & & \sum_{(i,j) \in S_=} (g_{ij}(G) - a_{ij}^2)^2 \\
        & \text{subject to}
        && g_{ij}(G) \leq \overline{a}_{ij}^2 \text{ for } (i,j) \in S_{\leq},\\
        &&& g_{ij}(G) \geq \underline{a}_{ij}^2 \text{ for } (i,j) \in S_{\geq},\\
        &&& g_{ij}(G) \leq \overline{a}_{ij}^2-\varepsilon \text{ for } (i,j) \in S_{<},\\
        &&& g_{ij}(G) \geq \underline{a}_{ij}^2+\varepsilon \text{ for } (i,j) \in S_{>},\\
        &&& \sum_{1 \leq i,j \leq n} G_{ij} = 0, \\
        &&& G \succeq 0.
    \end{aligned}
\end{equation}

In the formulation~\eqref{formulation:existence_unit_ball_graph_optimisation_trace}, we have replaced the strict inequalities $g_{ij}(G) < \overline{a}_{ij}^2$ (resp. $g_{ij}(G) > \underline{a}_{ij}^2$)  by $g_{ij}(G) \leq \overline{a}_{ij}^2-\varepsilon$ (resp. $g_{ij}(G) \geq \underline{a}_{ij}^2+\varepsilon$) for a sufficiently small epsilon to obtain non-strict inequalities. The constraint $\sum_{1 \leq i,j \leq n} G_{ij} = 0$ removes the freedom that comes from translations of configurations.  

Solving the above semidefinite program, gives us a Gram matrix $\hat{G}$. To find a point configuration in $\R^d$, we find the $d$ largest eigenvalues $\lambda_1,\ldots,\lambda_d \in \R$ of $\hat{G}$ and the corresponding eigenvectors $v_1,\ldots,v_d \in \R^n$. For $1 \leq i \leq n$, we define 
\begin{equation*}
    x_i = \left( \sqrt{\lambda_1} \cdot v_{1,i},\ldots,\sqrt{\lambda_d} \cdot v_{d,i}  \right).
\end{equation*}
This method is called classical multidimensional scaling (cMDS)~\citep{cox2008multidimensional}. If the optimal Gram matrix $\hat{G}$ found by the semidefinite program has rank at most $d$, the point configuration constructed using cMDS satisfies exactly the lower and upper bounds for distances. If the rank of the Gram matrix is greater than $d$, then the realisation in $\R^d$ might not satisfy all the bounds and we obtain an approximate solution.  

\begin{remark}
    Adding $\lambda \cdot \tr(G)$ or $-\lambda \cdot \tr(G)$ to the objective function has been suggested as a heuristic to obtain a low-rank Gram matrix.  
    The idea behind adding $-\lambda \cdot \tr(G)$ to the objective is that, in the case of noisy distances, the points tend to crowd towards the centre and maximising trace spreads out the locations of the beads~\citep{biswas2006semidefinite2}. An alternative approach is to minimise $\tr(G)$, since trace is the convex envelope of the rank function in the unit ball. Minimising trace instead of rank was suggested by \citet{fazel2001rank} and it has been applied to the 3D genome reconstruction setting from \citet{belyaeva2021identifying}. In our experiments, modifying the objective function did not improve results, so we used the original formulation~\eqref{formulation:existence_unit_ball_graph_optimisation_trace}. The term  $\lambda \cdot \tr(G)$ can be added in our implementation by giving $\lambda \in \R$ as an option.
\end{remark}

As the final step,  we solve the optimisation problem
\begin{align} \label{gradient_descent_objective}
    \min &\left(\sum_{(i,j) \in S_{\leq} \cup S_<} (\|x_i -x_j \|-\overline{a}_{ij})_+^2 + \sum_{(i,j) \in S_{\geq} \cup S_>}  (\|x_i -x_j \|-\underline{a}_{ij})_-^2 \right.\nonumber\\
    &\left.+ \sum_{(i,j) \in S_=}  (\|x_i -x_j \|-a_{ij})^2 \right)
\end{align}
using gradient descent initialised at the solution found by the multidimensional scaling. In the above formulation, $\alpha_+ = \max(0,\alpha)$ and $\alpha_- = \max(0,-\alpha)$. This step is the same as by \citet{leung2010sdp}. Even if the $n$-dimensional solution corresponding to the optimal Gram matrix satisfies the constraints~\eqref{eqn:constraints}, the $d$-dimensional solution obtained by cMDS might not. The goal of this step is to improve the $d$-dimensional solution.

To summarise, our reconstruction algorithm performs the following three steps:
\begin{enumerate}
    \item\label{it:alg:sdp} Running the SDP program~\eqref{formulation:existence_unit_ball_graph_optimisation_trace}. 
    \item\label{it:alg:gram} Using classical multidimensional scaling to obtain a 3D solution from the optimal Gram matrix obtained in step \ref{it:alg:sdp}.
    \item\label{it:alg:local} Minimising the objective~\eqref{gradient_descent_objective}
    using gradient descent with initialisation equal to the solution obtained in step \ref{it:alg:gram}. 
\end{enumerate}

\subsection{Experiments} \label{sec:experiments}

We run step \ref{it:alg:sdp} using \texttt{cvx} with the default solver \texttt{SDPT3}.
For step \ref{it:alg:local}, we use the \texttt{fminunc} command in \texttt{MATLAB}. The algorithm used for local optimisation is quasi-Newton. We use two measures of dissimilarity to study the quality of the reconstruction: 
\begin{enumerate}
    \item The modification of the objective function~\eqref{gradient_descent_objective} where we only keep the inequalities corresponding to the HiC matrix:
    \begin{equation} \label{violation}
        \min \sum_{(i,j): F(i,j)=1} (\|x_i -x_j \|-1)_+^2 + \sum_{(i,j): F(i,j)=0}  (\|x_i -x_j \| - 1)_-^2. 
    \end{equation}
    We call the value corresponding to the function~(\ref{violation}) the violation.
    \item The Procrustes distance between the original structure and the reconstructed structure. The Procrustes transformation finds the affine transformation that minimises the sum of squared distances between two ordered point sets. The Procrustes distance is the sum of squared distances between the aligned point sets divided by the scale of the first point set. We do allow scaling in the Procrustes algorithm.
\end{enumerate}
The Procrustes distance can be computed only if the original structure is known, e.g.\ for synthetic data sets. Moreover, it is comparing the reconstruction with the original structure, which is more than comparing Hi-C matrices. Therefore, we mostly focus on the first measure of dissimilarity.

\textbf{Synthetic data.} Our synthetic data set consists of random walks with $60$ points, where each step of the random walk is generated from the 3D sphere (all steps are of length one). A threshold $r$ is used to construct the Hi-C matrix. If the distance between two points is at most $r$, then the corresponding value in the Hi-C matrix is set to be 1, and otherwise it is set to be 0. For all models, we assume that we know the distance between neighbouring beads on a chromosome, which in the experiments is equal to one. 

\begin{figure}[ht]
     \centering
     \begin{subfigure}[b]{0.3\textwidth}
         \centering
         \includegraphics[width=\textwidth]{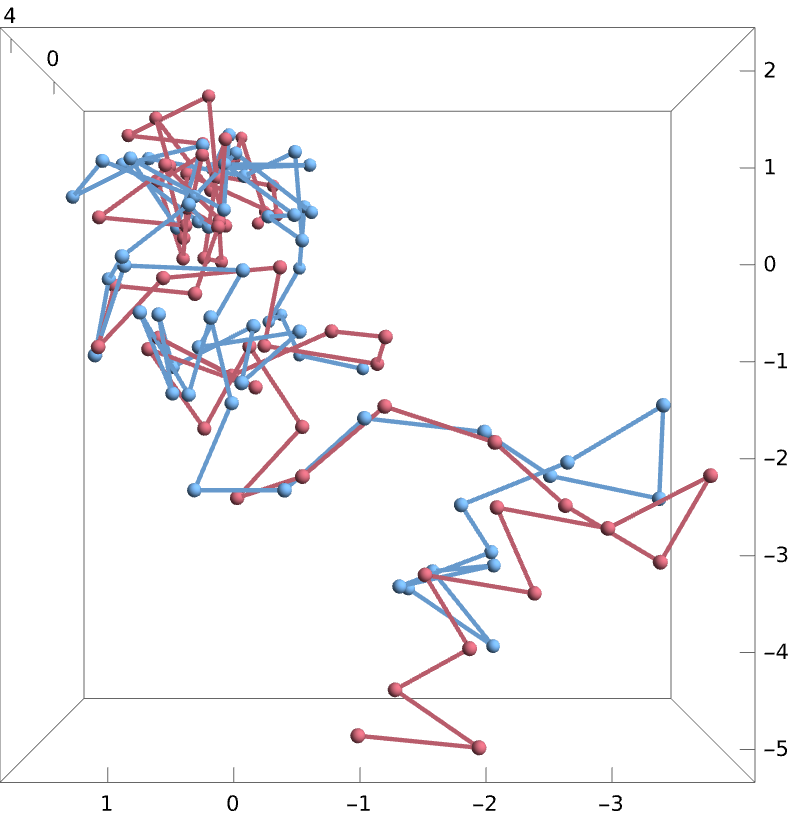}
         \caption{Unit ball}
         \label{fig:reconstruction-ub}
     \end{subfigure}
     \qquad
     \begin{subfigure}[b]{0.3\textwidth}
         \centering
         \includegraphics[width=\textwidth]{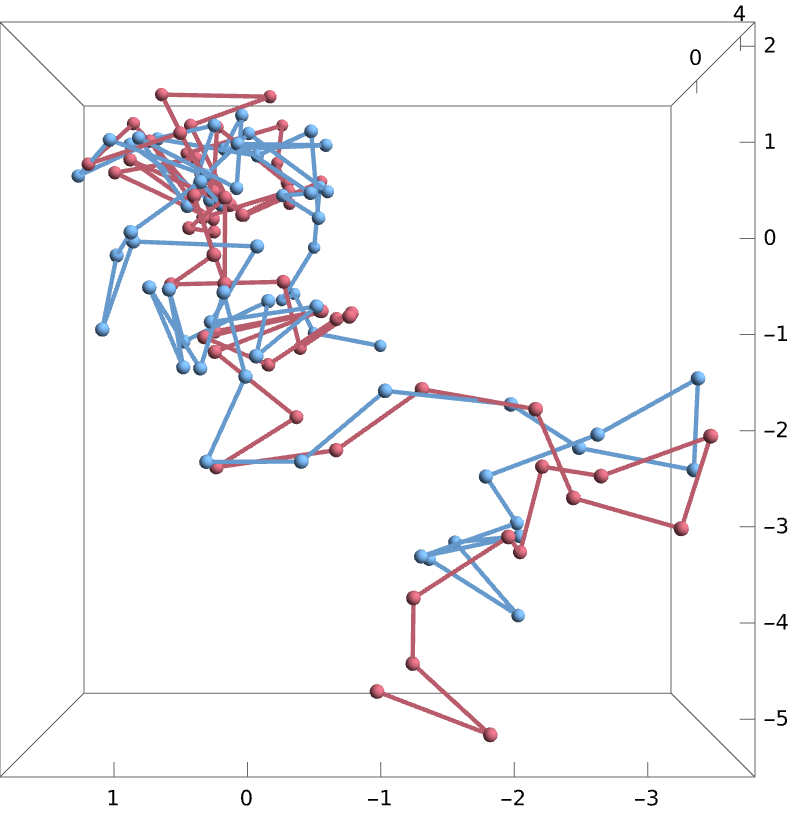}
         \caption{Marble}
         \label{fig:reconstruction-pg}
     \end{subfigure}\\[2ex]     
     \begin{subfigure}[b]{0.3\textwidth}
         \centering
         \includegraphics[width=\textwidth]{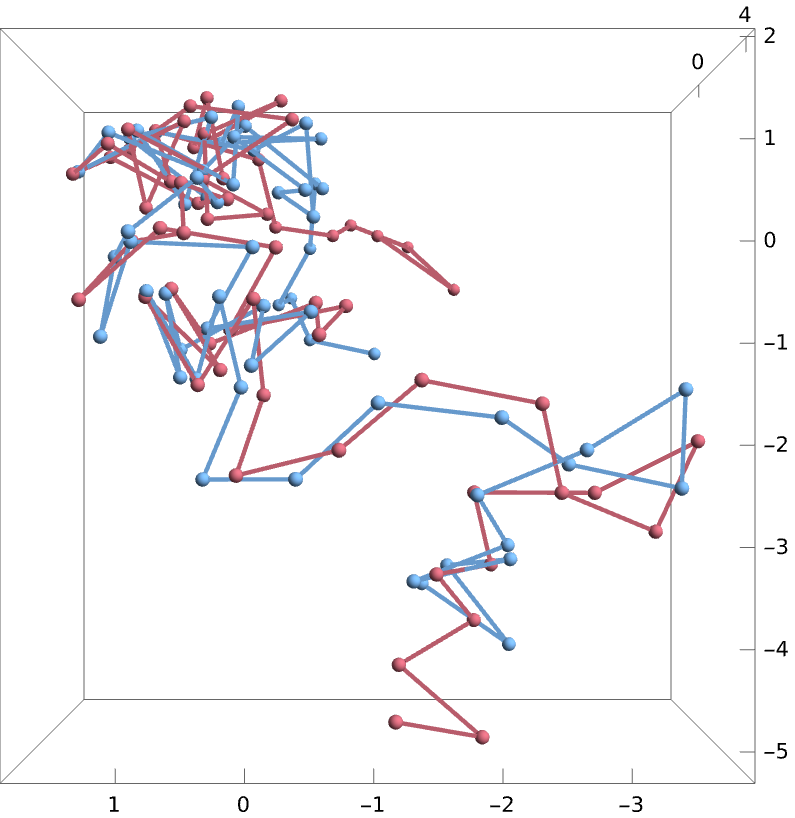}
         \caption{2\%}
         \label{fig:reconstruction-ei2}
     \end{subfigure}
     \quad
     \begin{subfigure}[b]{0.3\textwidth}
         \centering
         \includegraphics[width=\textwidth]{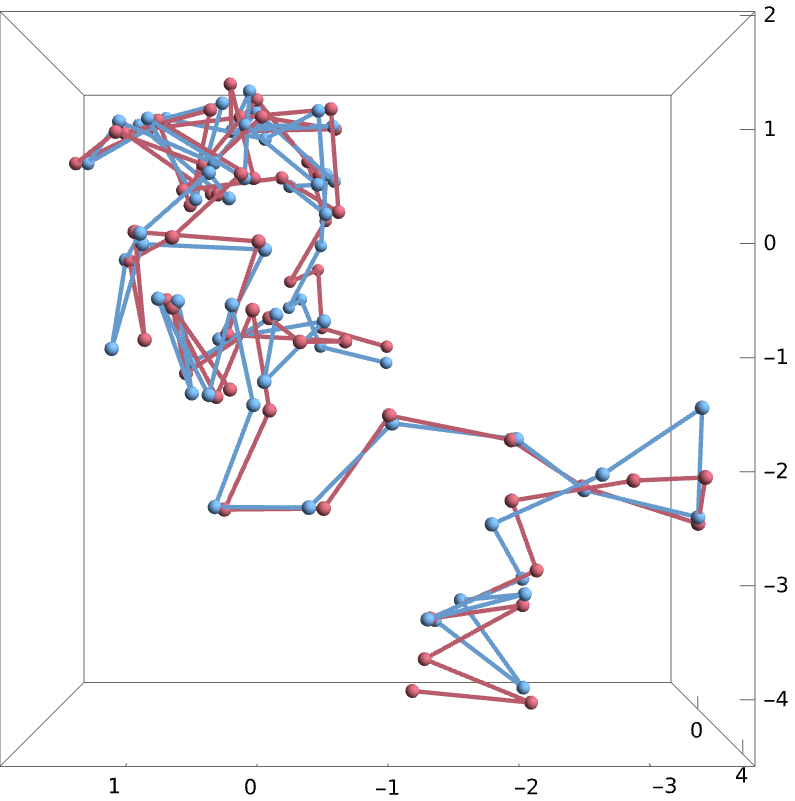}
         \caption{5\%}
         \label{fig:reconstruction-ei5}
     \end{subfigure}
     \quad
     \begin{subfigure}[b]{0.3\textwidth}
         \centering
         \includegraphics[width=\textwidth]{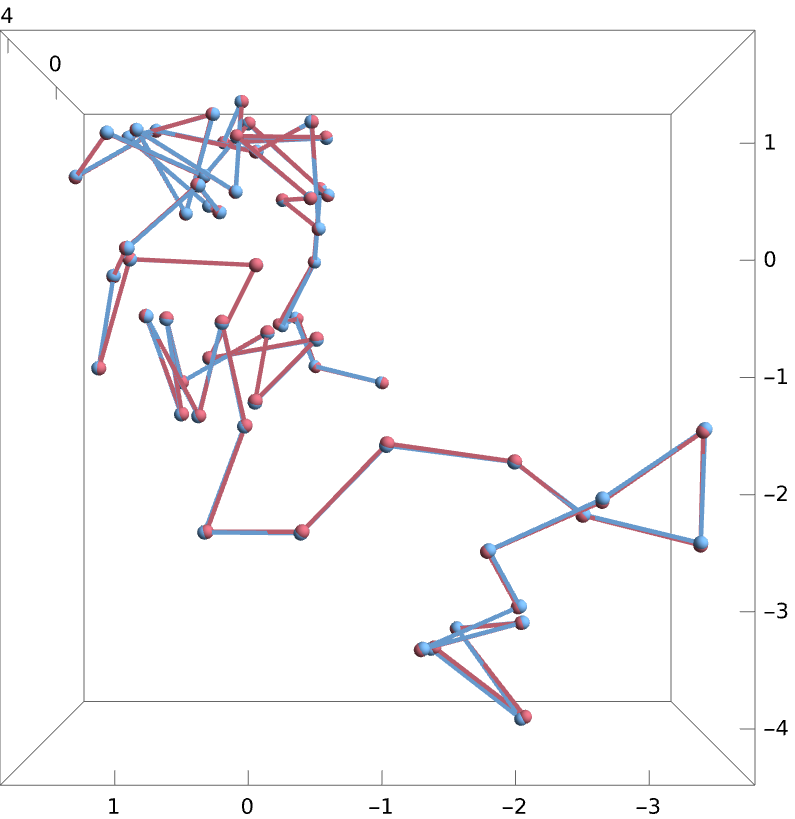}
         \caption{10\%}
         \label{fig:reconstruction-ei10}
     \end{subfigure}
    \caption{Examples of reconstructions obtained using our reconstruction method using the backbone information that the distances between neighbouring beads are equal to 1: \labelcref{fig:reconstruction-ub} Unit ball model, \labelcref{fig:reconstruction-pg} Marble graph model, \labelcref{fig:reconstruction-ei2,fig:reconstruction-ei5,fig:reconstruction-ei10} Equations and inequalities model with 2\%, 5\% and 10\% of edge lengths known. The chromosome has 60 beads generated from the normalised Brownian motion model and the threshold for a contact is $r=3$. The backbone of the original data is shown in blue, while red shows the reconstruction.}
    \label{fig:example_reconstructions}
\end{figure}

For each value $r \in \{1,2,3,4,5,6\}$, we generated 50 random walks with 60 vertices. 
\Cref{fig:example_reconstructions} shows example reconstructions using our reconstruction algorithm for unit ball, marble graph and models with equations and inequalities for the threshold $r=3$. 
In \Cref{fig:violationhic}, we show the median violation~(\ref{violation}) for the different models on a log-scale for thresholds $r$. 
Some results are numerically 0 which yields the extra line in the graphic. 
In \Cref{fig:procrustes}, we compare the Procrustes distance for these models.

In all three figures, the equations and inequalities model with 10\% of edge lengths known gives very good reconstructions. The same model with 5\% of edge lengths known also gives a relatively good reconstruction in~\Cref{fig:reconstruction-ei5} and small Procrustes distance, but it gives the worst violation out of all the models. We explored the performance of the violation for the equations and inequalities model with different percentage of edge lengths and observed that the violation increases when going from 0\% of edges known to about 5\% edges known, and decreases after that. One explanation of this could be that with no edge lengths known, i.e.\ in the unit ball model, it is possible to obtain a 3D structure with low violation although it might be different from the original structure. When adding more edge lengths, initially the violation can get worse, although the Procrustes distance decreases, and if enough edges are added then both violation and Procrustes distance decrease. The other three models get the high-level details of the reconstructions correct, but details are not constructed correctly. These observations are consistent with the identifiability results in the earlier sections, where we observed that a certain number of equality constraints are needed to obtain finitely many or a unique reconstruction.

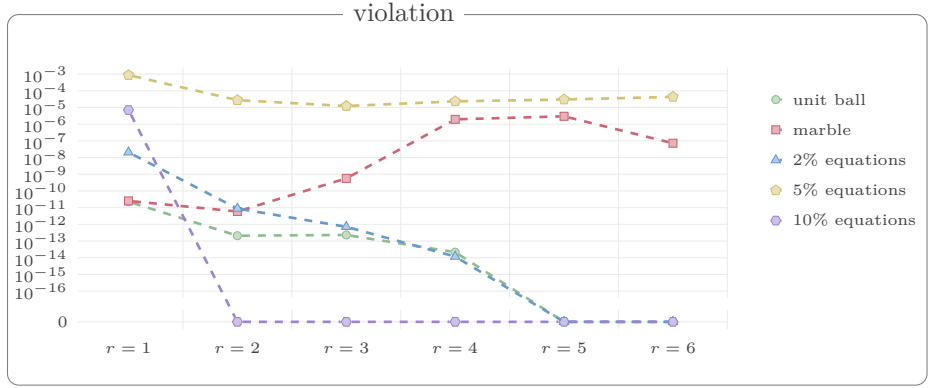
\begin{figure}[ht]
    \centering
    \begin{tikzpicture}[scale=0.4,xscale=0.8]
    	\draw[colfg!50!colbg,rounded corners=0.15cm] (-3,-11.9393) rectangle (35,0.310657);
    	\node[hlabelsty,fill=white] at (13.375,0.310657) {violation\vphantom{hg}};
    	\foreach \n [count=\ii,evaluate=\ii as \p using 4*0.75/2+(6*0.75*(\ii-1)+0.5)] in {1,2,...,6}
    	{
    		\node[below,alabelsty] at (\p,-10.1893) {$r=\n$};
    	}
    	\foreach \n [count=\ii,evaluate=\ii as \p using 5*0.75+(6*0.75*(\ii-1)+0.5)] in {1,2,...,6}
    	{
    		\draw[bline] (\p,-1.43934)--(\p,-9.04696);
    		\draw[bline] (\p,-9.44696)--(\p,-10.0893);
    	}
    	\foreach \y/\l in {-1.65604/$10^{-3}$,-2.20805/$10^{-4}$,-2.76006/$10^{-5}$,-3.31207/$10^{-6}$,-3.86408/$10^{-7}$,-4.4161/$10^{-8}$,-4.96811/$10^{-9}$,-5.52012/$10^{-10}$,-6.07213/$10^{-11}$,-6.62414/$10^{-12}$,-7.17615/$10^{-13}$,-7.72817/$10^{-14}$,-8.28018/$10^{-15}$,-8.83219/$10^{-16}$}
    	{
    		\draw[bline] (0,\y)--(-0.1,\y) node[left,alabelsty] {\l};
    		\draw[bline] (0,\y)--(26.75,\y);
    	}
    	\draw[bline] (26.75,-9.84696)--(-0.1,-9.84696) node[left,alabelsty] {0};
    	\begin{scope}[xshift=1.5cm]
    		\foreach \median [count=\ii,evaluate=\ii as \i using 6*0.75*(\ii-1)+0.5] in {-5.87423, -6.99971, -6.97557, -7.544, -9.84696, -9.84696}
    		{
    			\node[mediannode1,draw=col1,fill=col1!50!colbg] (a\ii) at (\i,\median) {};
    		}
    	\end{scope}
    	\begin{scope}[xshift=1.5cm]
    		\foreach \median [count=\ii,evaluate=\ii as \i using 6*0.75*(\ii-1)+0.5] in {-5.84943, -6.19949, -5.10501, -3.15311, -3.05228, -3.94284}
    		{
    			\node[mediannode2,draw=col2,fill=col2!50!colbg] (b\ii) at (\i,\median) {};
    		}
    	\end{scope}
    	\begin{scope}[xshift=1.5cm]
    		\foreach \median [count=\ii,evaluate=\ii as \i using 6*0.75*(\ii-1)+0.5] in {-4.24247, -6.10172, -6.70558, -7.68934, -9.84696, -9.84696}
    		{
    			\node[mediannode3,draw=col3,fill=col3!50!colbg] (c\ii) at (\i,\median) {};
    		}
    	\end{scope}
    	\begin{scope}[xshift=1.5cm]
    		\foreach \median [count=\ii,evaluate=\ii as \i using 6*0.75*(\ii-1)+0.5] in {-1.68934, -2.51812, -2.71288, -2.55726, -2.49706, -2.40953}
    		{
    			\node[mediannode4,draw=col4,fill=col4!50!colbg] (d\ii) at (\i,\median) {};
    		}
    	\end{scope}
    	\begin{scope}[xshift=1.5cm]
    		\foreach \median [count=\ii,evaluate=\ii as \i using 6*0.75*(\ii-1)+0.5] in {-2.84434, -9.84696, -9.84696, -9.84696, -9.84696, -9.84696}
    		{
    			\node[mediannode5,draw=col5,fill=col5!50!colbg] (e\ii) at (\i,\median) {};
    		}
    	\end{scope}
    	\begin{scope}
    	\foreach \r [remember=\r as \ro (initially 1)] in {2,3,4,5,6}
    		{
    			\draw[medianline,draw=col1] (a\ro)--(a\r);
    			\draw[medianline,draw=col2] (b\ro)--(b\r);
    			\draw[medianline,draw=col3] (c\ro)--(c\r);
    			\draw[medianline,draw=col4] (d\ro)--(d\r);
    			\draw[medianline,draw=col5] (e\ro)--(e\r);
    		}
    	\end{scope}
    	\begin{scope}[shift={(28.75,-1.5)}]
    		\foreach \l [count=\i] in {unit ball,marble,2\% equations,5\% equations,10\% equations}
    		{
    			\node[mediannode\i,draw=col\i,fill=col\i!50!colbg] (l\i) at (0,-\i) {};
                \draw[medianline,col\i] (l\i)--++(0.25,0) node[alabelsty,anchor=west] {\l};
    		}
    	\end{scope}
    \end{tikzpicture}
    \caption{Medians of the violation to the HiC-matrix for unit ball model, marble graph model and equations and inequalities model with 2\%, 5\% and 10\% of edge lengths known. The medians are plotted in log-scale. The chromosomes have 60 beads generated from the normalised Brownian motion model with backbone. The violation is shown for different thresholds $r$.}
    \label{fig:violationhic}
\end{figure}

\begin{figure}[ht]
    \centering
    \begin{tikzpicture}[scale=0.4,xscale=0.8]
    	\draw[colfg!50!colbg,rounded corners=0.15cm] (-3,-2) rectangle (35,8.2316);
    	\node[hlabelsty,fill=white] at (13.375,8.2316) {procrustes distance\vphantom{hg}};
    	\foreach \n [count=\ii,evaluate=\ii as \p using 4*0.75/2+(6*0.75*(\ii-1)+0.5)] in {1,2,...,6}
    	{
    		\node[below,alabelsty] at (\p,-0.3) {$r=\n$};
    	}
    	\foreach \n [count=\ii,evaluate=\ii as \p using 5*0.75+(6*0.75*(\ii-1)+0.5)] in {1,2,...,6}
    	{
    		\draw[bline] (\p,0)--(\p,-0.4);
    	}
    	\foreach \l [evaluate=\l as \y using \l*19.4738] in {0., 0.04, 0.08, 0.12, 0.16, 0.2, 0.24, 0.28, 0.32}
    	{
    		\draw[bline] (0,\y)--(-0.1,\y) node[left,alabelsty] {\l};
    		\draw[bline] (0,\y)--(26.75,\y);
    	}
    	\begin{scope}[xshift=0.cm]
    		\foreach \min/\q/\median/\p/\max [count=\ii,evaluate=\ii as \i using 6*0.75*(\ii-1)+0.5] in {1.20948/2.96518/3.98981/5.75632/11.4975,0.36146/0.95731/1.80182/2.95659/7.62653,0.15049/0.52101/0.78731/1.49864/3.33727,0.32993/0.78516/0.99305/1.48431/5.29215,0.57639/1.12213/1.62627/2.34529/9.80356,0.88941/1.78246/2.50433/4.23291/15.3411}
    		{
    			\draw[boxplot,draw=col1,fill=col1!50!colbg] (\i-0.25,\q) rectangle (\i+0.25,\p);
    			\draw[boxplot,draw=col1,fill=col1!50!colbg] (\i-0.25,\median)--(\i+0.25,\median);
    		}
    	\end{scope}
    	\begin{scope}[xshift=0.75cm]
    		\foreach \min/\q/\median/\p/\max [count=\ii,evaluate=\ii as \i using 6*0.75*(\ii-1)+0.5] in {1.50587/3.04531/4.11977/6./11.5701,0.42929/1.07338/1.84543/3.14701/7.66943,0.25889/0.90965/1.35322/2.1906/3.49579,0.48989/1.45102/1.80353/2.37668/6.64076,0.87959/1.90308/2.70125/3.44042/9.79175,1.70472/2.70295/3.71249/5.25357/17.1662}
    		{
    			\draw[boxplot,draw=col2,fill=col2!50!colbg] (\i-0.25,\q) rectangle (\i+0.25,\p);
    			\draw[boxplot,draw=col2,fill=col2!50!colbg] (\i-0.25,\median)--(\i+0.25,\median);
    		}
    	\end{scope}
    	\begin{scope}[xshift=1.5cm]
    		\foreach \min/\q/\median/\p/\max [count=\ii,evaluate=\ii as \i using 6*0.75*(\ii-1)+0.5] in {0.78706/1.71335/2.30799/3.2915/5.55312,0.13471/0.566/0.91785/1.59177/3.69013,0.06797/0.33047/0.52153/0.8822/2.06512,0.20596/0.49011/0.68106/1.00602/2.90097,0.14418/0.66992/1.03969/1.96585/8.16903,0.28895/0.86702/1.42457/3.07789/12.7252}
    		{
    			\draw[boxplot,draw=col3,fill=col3!50!colbg] (\i-0.25,\q) rectangle (\i+0.25,\p);
    			\draw[boxplot,draw=col3,fill=col3!50!colbg] (\i-0.25,\median)--(\i+0.25,\median);
    		}
    	\end{scope}
    	\begin{scope}[xshift=2.25cm]
    		\foreach \min/\q/\median/\p/\max [count=\ii,evaluate=\ii as \i using 6*0.75*(\ii-1)+0.5] in {0.05464/0.44913/0.66303/1.22347/3.15903,0.03849/0.09981/0.14897/0.30683/2.47264,0.01965/0.07093/0.15708/0.32998/0.85908,0.03447/0.16921/0.33651/0.51352/2.84682,0.02045/0.20366/0.37188/1.00768/4.82647,0.10896/0.41018/0.63202/1.89871/10.0211}
    		{
    			\draw[boxplot,draw=col4,fill=col4!50!colbg] (\i-0.25,\q) rectangle (\i+0.25,\p);
    			\draw[boxplot,draw=col4,fill=col4!50!colbg] (\i-0.25,\median)--(\i+0.25,\median);
    		}
    	\end{scope}
    	\begin{scope}[xshift=3.cm]
    		\foreach \min/\q/\median/\p/\max [count=\ii,evaluate=\ii as \i using 6*0.75*(\ii-1)+0.5] in {0./0.00033/0.00159/0.02127/0.1508,0./0.00015/0.00089/0.00543/0.1291,0./0.00011/0.00072/0.00376/0.07626,0.00001/0.00019/0.00223/0.00782/0.08843,0./0.00034/0.00449/0.03925/0.43069,0.00005/0.00066/0.01867/0.07299/0.45026}
    		{
    			\draw[boxplot,draw=col5,fill=col5!50!colbg] (\i-0.25,\q) rectangle (\i+0.25,\p);
    			\draw[boxplot,draw=col5,fill=col5!50!colbg] (\i-0.25,\median)--(\i+0.25,\median);
    		}
    	\end{scope}
    	\begin{scope}[shift={(28.75,6.1158)}]
    		\foreach \l [count=\i] in {unit ball,marble,2\% equations,5\% equations,10\% equations}
    		{
    			\draw[line width=2pt,col\i] (0,-\i)--++(0.25,0) node[alabelsty,anchor=west] {\l};
    		}
    	\end{scope}
    \end{tikzpicture}
    \caption{Medians and 25/75\%-quantiles of the Procrustes distance (including scaling) for unit ball model, marble graph model and equations and inequalities model with 2\%, 5\% and 10\% of edge lengths known. The chromosomes have 60 beads generated from the normalised Brownian motion model with backbone. The violation is shown for different thresholds $r$.}
    \label{fig:procrustes}
\end{figure}
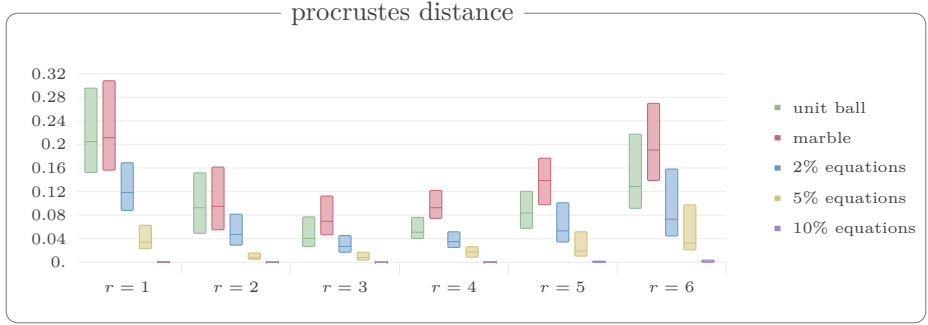

To have some idea of how good the results are in relation to the number of beads we computed the Procrustes distance allowing scaling for 50 repetitions of the experiment with simulated data for 10 to 150 beads. The data was generated using $r=3$. We also compare the results with ShRec3D algorithm from \citet{lesne20143d} that combines shortest path computations with multidimensional scaling.
The result can be seen in \Cref{fig:proccompare}. For these small examples our methods turn out to have lower Procrustes distance than the one obtained from ShRec3D.

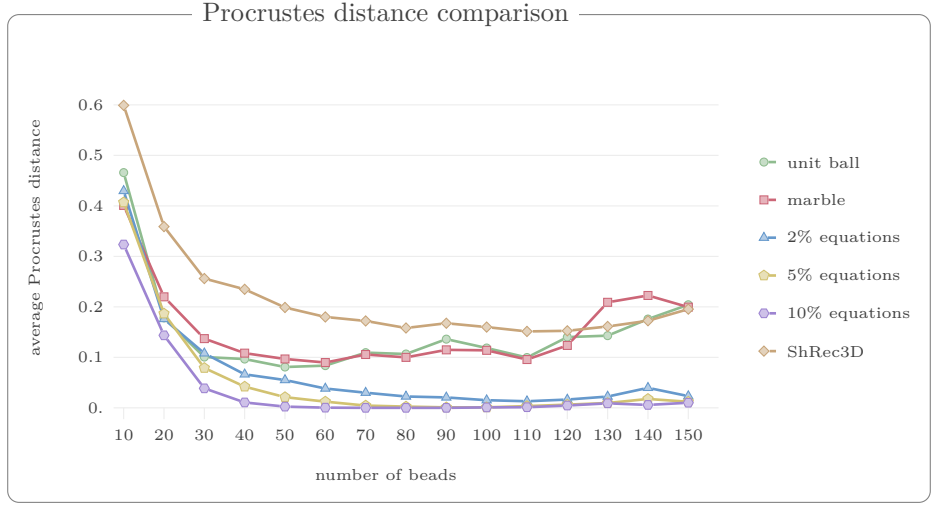
\begin{figure}[ht]
    \centering
    \begin{tikzpicture}[xscale=0.8]
    	\draw[colfg!50!colbg,rounded corners=0.15cm] (-1.25,-1.25) rectangle (14,5.2);
    	\node[hlabelsty,fill=white] at (5,5.2) {Procrustes distance comparison\vphantom{hg}};
    	\foreach \l [evaluate=\l as \y using 0.13355*\l*50] in {0., 0.1, 0.2, 0.3, 0.4, 0.5, 0.6}
    	{
    		\draw[bline] (10.5,\y)--(0.5,\y) node[alabelsty,anchor=east] {\l};
    	}
    	\node[alabelsty,anchor=south,rotate=90] at (-0.5,0.01*225) {average Procrustes distance};
    	\foreach \t [count=\i,evaluate=\i as \x using 0.66667*\i] in {10, 20, 30, 40, 50, 60, 70, 80, 90, 100, 110, 120, 130, 140, 150}
    	{
    		\draw[bline] (\x,0)--(\x,-0.15) node[alabelsty,anchor=north] {\t};
    	}
    	\node[alabelsty,anchor=north] at (5,-0.675) {number of beads};
    
    	\foreach \ll [count=\m] in {%
    		{3.11069, 1.18952, 0.67184, 0.64534, 0.53986, 0.55828, 0.72992, 0.71181, 0.90845, 0.79096, 0.6654, 0.93535, 0.95465, 1.17411, 1.3632},%
    		{2.67668, 1.46788, 0.91541, 0.72327, 0.64624, 0.5992, 0.7053, 0.66733, 0.76679, 0.76052, 0.63881, 0.82521, 1.39555, 1.48584, 1.33165},%
    		{2.86805, 1.17504, 0.72087, 0.44281, 0.36774, 0.25553, 0.20011, 0.15002, 0.13769, 0.10095, 0.0871, 0.10954, 0.14921, 0.26439, 0.15494},%
    		{2.71697, 1.24595, 0.52701, 0.28045, 0.14123, 0.08209, 0.02977, 0.01644, 0.0065, 0.0064, 0.0231, 0.04214, 0.06377, 0.11748, 0.08034},%
    		{2.15952, 0.95903, 0.25799, 0.0712, 0.01627, 0.00245, 0.00015, 0.0001, 0.00001, 0.00594, 0.00854, 0.0306, 0.05998, 0.03849, 0.06775},%
    		{4.00008, 2.39667, 1.70972, 1.56752, 1.32724, 1.2019, 1.14908, 1.05594, 1.11868, 1.06796, 1.01045, 1.01859, 1.07601, 1.15174, 1.30511}%
    	}
    	{
    		\foreach \y [count=\i,evaluate=\i as \x using 0.66667*\i]in \ll
    		{
    			\node[mediannode\m,draw=col\m,fill=col\m!50!colbg] (v\m\i) at (\x,\y) {};
    		}
    		\foreach \i [remember=\i as \io (initially 1)] in {2,3,...,15}
    		{
    			\draw[timeline,draw=col\m] (v\m\i)--(v\m\io);
    		}
    	}
    	\begin{scope}[shift={(11.25,3.75)}]
    		\foreach \l [count=\i] in {unit ball,marble,2\% equations,5\% equations,10\% equations,ShRec3D}
    		{
    			\node[mediannode\i,draw=col\i,fill=col\i!50!colbg] (l\i) at (0,-0.5*\i) {};
                \draw[timeline,draw=col\i] (l\i)--++(0.2,0) node[alabelsty,anchor=west] {\l};
    		}
    	\end{scope}
    \end{tikzpicture}
    \caption{Comparison of the Procrustes distance. The value shows the average of all the Procrustes distances of 50 repetitions for different methods.}
    \label{fig:proccompare}
\end{figure}

In \Cref{fig:timings} we compare the time needed for different number of beads. The figure shows the average time in seconds over 50 repetitions of the experiment with simulated data.  A limitation of the SDP based methods is that they are relatively slow, so we consider up to 150 beads. In this case, ShRec3D algorithm is much faster. In summary, the SDP based methods using models considered in this paper give better reconstructions but they are considerably slower compared to ShRec3D.

\begin{figure}[ht]
    \centering
    \begin{tikzpicture}[xscale=0.8]
    	\draw[colfg!50!colbg,rounded corners=0.15cm] (-1.25,-1.25) rectangle (14,5.2);
    	\node[hlabelsty,fill=white] at (5,5.2) {time comparison\vphantom{hg}};
    	\foreach \l [evaluate=\l as \y using 0.00029*\l*50] in {0, 40, 80, 120, 160, 200, 240, 280}
    	{
    		\draw[bline] (10.5,\y)--(0.5,\y) node[alabelsty,anchor=east] {\l};
    	}
    	\node[alabelsty,anchor=south,rotate=90] at (-0.5,0.01*225) {average time (seconds)};
    	\foreach \t [count=\i,evaluate=\i as \x using 1.*\i] in {10, 20, 30, 40, 50, 60, 70, 80, 90, 100}
    	{
    		\draw[bline] (\x,0)--(\x,-0.15) node[alabelsty,anchor=north] {\t};
    	}
    	\node[alabelsty,anchor=north] at (5,-0.675) {number of beads};
    
    	\foreach \ll [count=\m] in {%
    		{0.01556, 0.03009, 0.0566, 0.10763, 0.16861, 0.30367, 0.50893, 0.81387, 1.30642, 1.98083},%
    		{0.03704, 0.11326, 0.21345, 0.34702, 0.57836, 1.00955, 1.44701, 2.05252, 2.81617, 4.0519},%
    		{0.01544, 0.03552, 0.06678, 0.12748, 0.23069, 0.43662, 0.73439, 1.23113, 1.98822, 2.94694},%
    		{0.01409, 0.03626, 0.08197, 0.16032, 0.28764, 0.57213, 0.96286, 1.36487, 2.41736, 3.96355},%
    		{0.02136, 0.08521, 0.15803, 0.20864, 0.32763, 0.66721, 1.10847, 1.73967, 2.64405, 3.95206},%
    		{0.00003, 0.00006, 0.00005, 0.00006, 0.00008, 0.0001, 0.00012, 0.00014, 0.00017, 0.00023}%
    	}
    	{
    		\foreach \y [count=\i,evaluate=\i as \x using 1.*\i]in \ll
    		{
    			\node[mediannode\m,draw=col\m,fill=col\m!50!colbg] (v\m\i) at (\x,\y) {};
    		}
    		\foreach \i [remember=\i as \io (initially 1)] in {2,3,...,10}
    		{
    			\draw[timeline,draw=col\m] (v\m\i)--(v\m\io);
    		}
    	}
    	\begin{scope}[shift={(11.25,3.75)}]
    		\foreach \l [count=\i] in {unit ball,marble,2\% equations,5\% equations,10\% equations,ShRec3D}
    		{
    			\node[mediannode\i,draw=col\i,fill=col\i!50!colbg] (l\i) at (0,-0.5*\i) {};
                \draw[timeline,draw=col\i] (l\i)--++(0.2,0) node[alabelsty,anchor=west] {\l};
    		}
    	\end{scope}
    \end{tikzpicture}
    \caption{Time comparison of the methods using backbone data on synthetic examples from before. The time is given on average and is shown in seconds for 50 repetitions.}
    \label{fig:timings}
\end{figure}
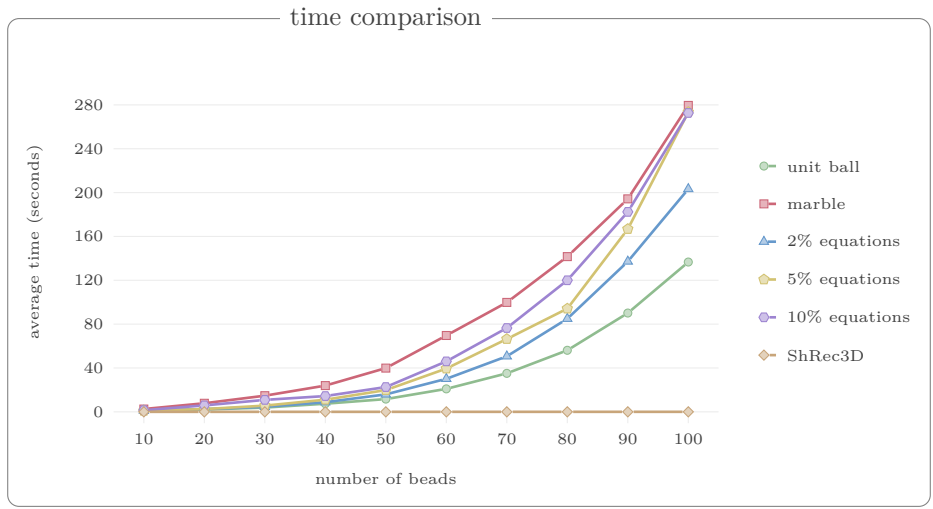

\textbf{Real data.} \citet{stevens20173d} combined imagining with a HiC protocol to obtain 3D structures for single G1 phase haploid mouse embryonic stem cells \citep[data accessible at NCBI GEO database][accession GSE80280]{edgar2002gene}.  We applied our reconstruction algorithm with unit ball model on Chromosome 1 in Sample GSM2219497 of this dataset. The chromosome consists of $\sim193 \cdot 10^6$ base pairs. Each bead in our model corresponds to one million base pairs. There is a contact between two beads if and only if there is at least one contact between base pairs in the corresponding beads. We assume that contacts
correspond to distances $r \leq 2$. Our data consists of 193 beads and 802 contacts. Note that this aggregation of data destroys possible backbone information. Transferring the information suitably to the aggregated data is the subject of further investigation. Here, we restrict attention to the case where the backbone information is lost.
The 3D reconstruction obtained using the unit ball model without backbone information is depicted in~\Cref{fig:realdatareconstruct}. \Cref{fig:realdatabackbone} shows the backbone of the reconstruction and \Cref{fig:realdatacontacts} the contacts in the reconstruction. Extended Data Figure 1d from \citet{stevens20173d} shows an analogous reconstruction at 100kb scale with contacts that are more than distance four apart coloured red. \Cref{fig:histogram} shows the distribution of distances in the reconstructed structure for pairs of beads that are in contact and not in contact as a histogram. The maximal distance of a contact in the reconstruction is $\sim2.37643$ and the minimal distance of a non-contact in the reconstruction is $\sim1.48434$. 

\begin{figure}[ht]
     \centering
     \begin{subfigure}[b]{0.45\textwidth}
         \centering
         \includegraphics[width=0.9\textwidth]{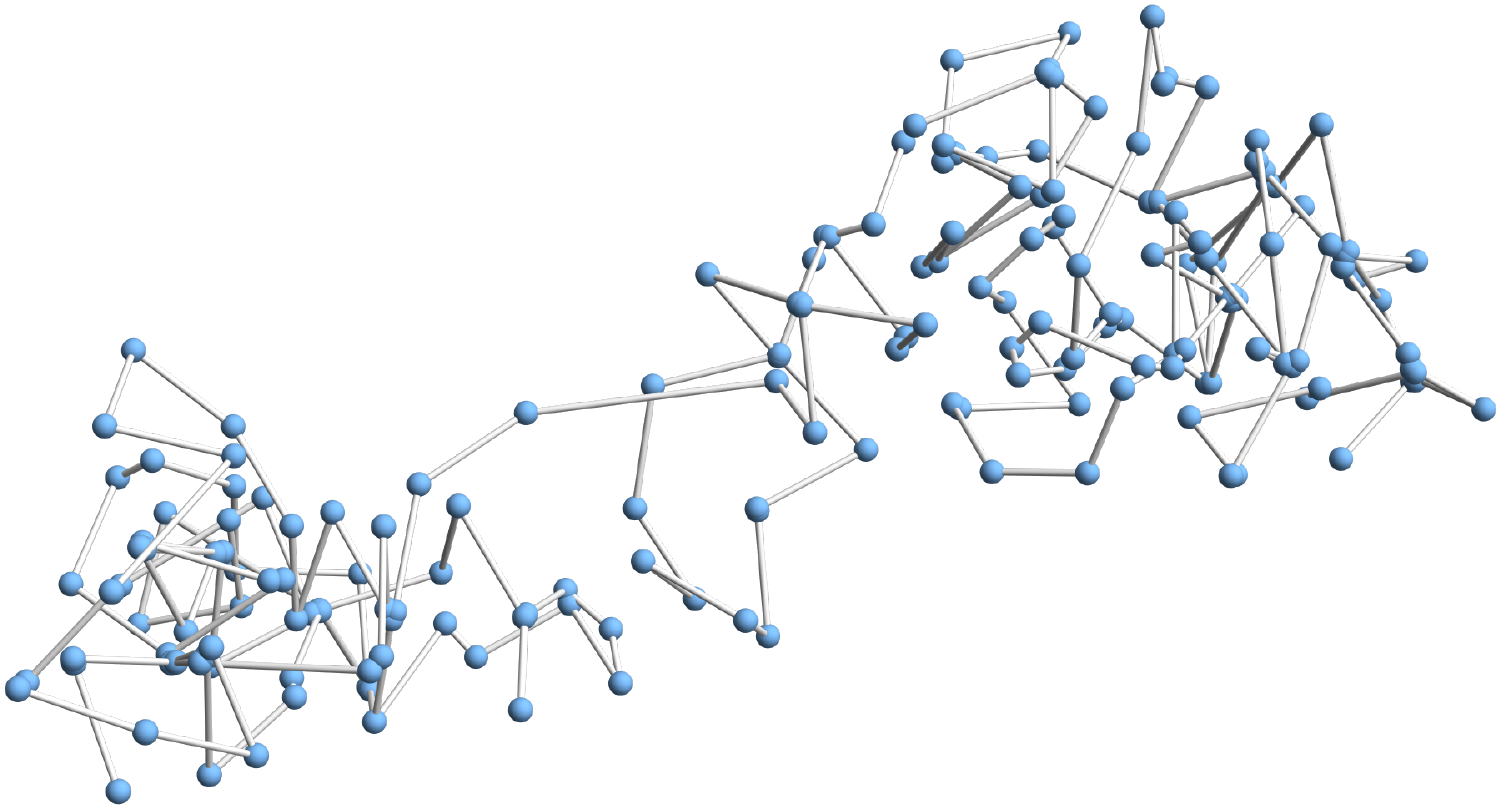}
         \caption{Reconstructed backbone}
         \label{fig:realdatabackbone}
     \end{subfigure}
     \hfill
     \begin{subfigure}[b]{0.45\textwidth}
         \centering
         \includegraphics[width=0.9\textwidth]{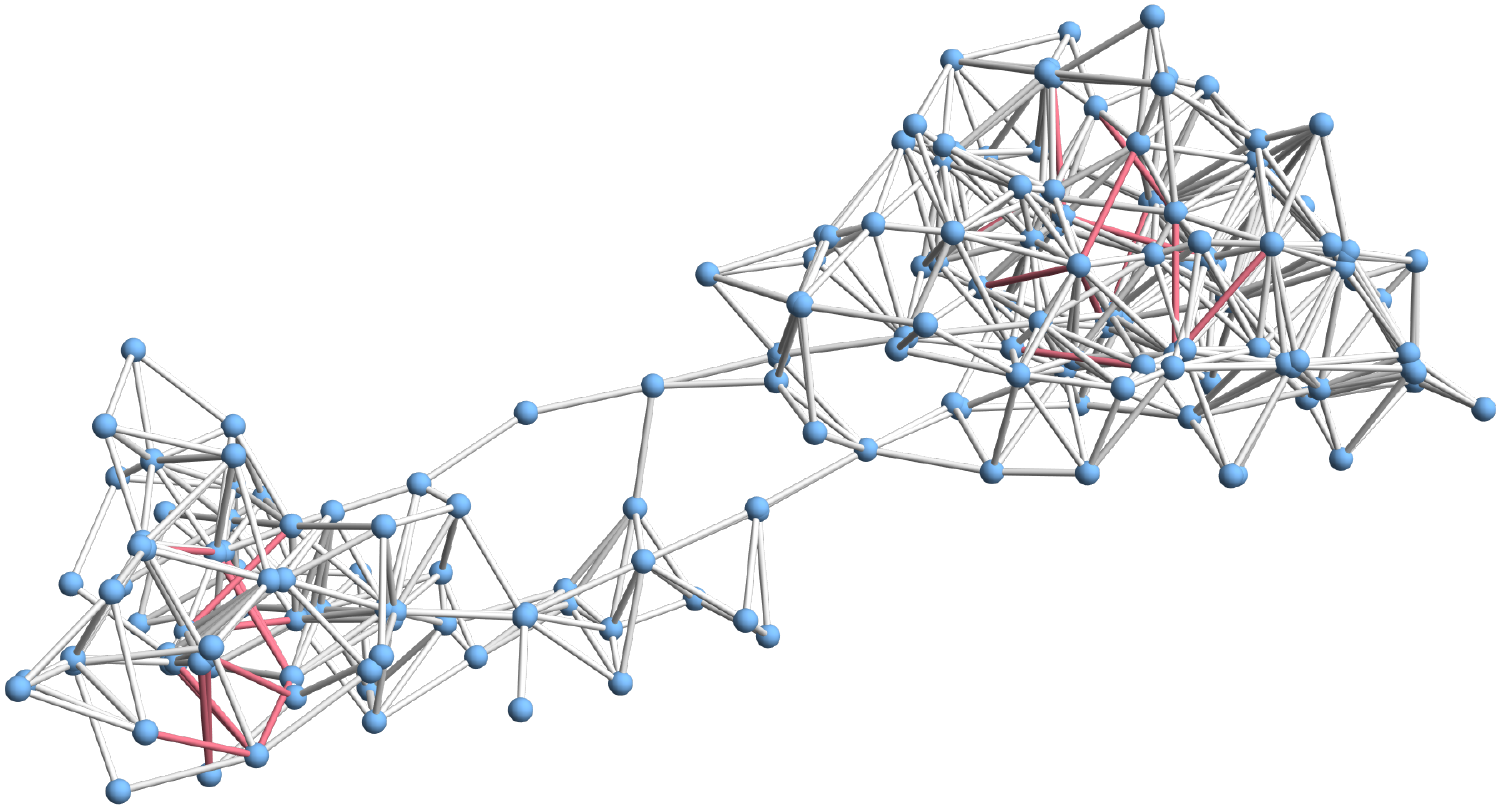}
         \caption{Reconstructed contacts}
         \label{fig:realdatacontacts}
     \end{subfigure}
     \caption{Reconstruction from a real data set using our reconstruction algorithm for the unit ball model without backbone information. (a) Backbone in the reconstruction. (b) Contacts in the reconstruction, i.e.\ edges are shown between beads that are in contact according to the aggregated HiC data. The edges which are too large by more than 10\% in the reconstruction are coloured in red.}
     \label{fig:realdatareconstruct}
\end{figure}

\begin{figure}[ht]
    \centering
    \begin{tikzpicture}[xscale=0.39]
        \draw[colfg!50!colbg,rounded corners=0.15cm] (-3.5,-1.25) rectangle (29.5,5.2);
	    \node[hlabelsty,fill=white] at (12,5.2) {edge lengths\vphantom{hg}};
        \fill[black!3,rounded corners] (0,0)rectangle(2,0.01*465);
        \foreach \l in {0,100,200,300,400}
        {
            \draw[bline] (24,0.01*\l)--(-0.5,0.01*\l) node[alabelsty,anchor=east] {\l};
        }
        \node[alabelsty,anchor=south,rotate=90] at (-2.25,0.01*225) {count};
        \foreach \t in {0,1,...,24}
        {
            \draw[bline] (\t,0)--(\t,-0.15) node[alabelsty,anchor=north] {\t};
        }
        \node[alabelsty,anchor=north] at (12,-0.675) {distance};
        \begin{scope}[xshift=25cm,yshift=4cm]
    		\foreach \l/\col [count=\i] in {edges/colG,non-edges/colO}
    		{
    			\draw[line width=2pt,\col] (0,-\i)--++(0.25,0) node[alabelsty,anchor=west] {\l};
    		}
        \end{scope}
        \foreach \l/\n in {0.2/44,0.4/29,0.6/26,0.8/29,1./23,1.2/28,1.4/33,1.6/49,1.8/78,2./179,2.2/260,2.4/24}
        {
            \fill[colG,opacity=0.5,draw=white,rounded corners=1] (\l-0.2,0) rectangle (\l,0.01*\n);
        }
        \foreach \l/\n in {1.6/2,1.8/18,2./330,2.2/340,2.4/191,2.6/182,2.8/180,3./220,3.2/269,3.4/419,3.6/446,3.8/461,4./455,4.2/330,4.4/285,4.6/296,4.8/301,5./369,5.2/347,5.4/351,5.6/308,5.8/237,6./245,6.2/238,6.4/222,6.6/234,6.8/236,7./238,7.2/212,7.4/195,7.6/161,7.8/142,8./182,8.2/176,8.4/163,8.6/159,8.8/174,9./156,9.2/133,9.4/118,9.6/134,9.8/148,10./177,10.2/153,10.4/127,10.6/167,10.8/149,11./127,11.2/137,11.4/154,11.6/160,11.8/145,12./159,12.2/157,12.4/160,12.6/127,12.8/135,13./145,13.2/185,13.4/205,13.6/150,13.8/178,14./164,14.2/181,14.4/145,14.6/170,14.8/187,15./177,15.2/199,15.4/182,15.6/142,15.8/174,16./180,16.2/159,16.4/169,16.6/179,16.8/142,17./140,17.2/148,17.4/153,17.6/164,17.8/148,18./121,18.2/117,18.4/102,18.6/95,18.8/117,19./86,19.2/93,19.4/91,19.6/62,19.8/47,20./58,20.2/58,20.4/46,20.6/43,20.8/43,21./28,21.2/27,21.4/17,21.6/16,21.8/16,22./9,22.2/9,22.4/10,22.6/1,22.8/1,23./3,23.2/5,23.4/2}
        {
            \fill[colO,opacity=0.5,draw=white,rounded corners=1] (\l-0.2,0) rectangle (\l,0.01*\n);
        }
    \end{tikzpicture}
    \caption{Histogram showing the edge lengths and non-edge-lengths of the real data reconstruction. The shaded area corresponds to the threshold below which contacts should occur. We see the actual edge lengths (green) and distances of non-edges (yellow) and their respective amounts.}
    \label{fig:histogram}
\end{figure}
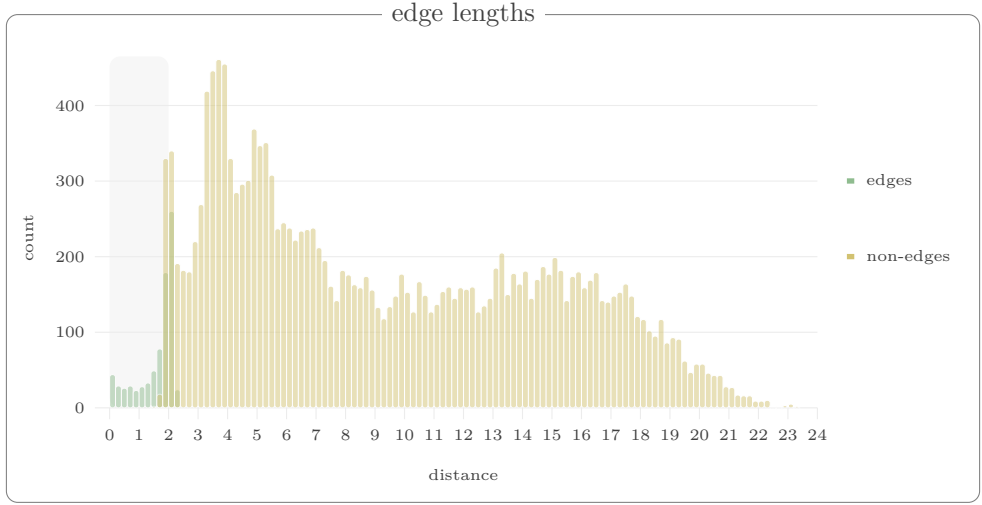

\section{Discussion}

In this article, we analysed the identifiability and reconstruction algorithms for single-cell 3D genome reconstruction in haploid organisms. The main contribution of the paper is introducing a correspondence between the biology of 3D genome reconstruction and rigidity theory. We study four different rigidity theoretic models corresponding to different types, or combinations, of biological measurements: unit ball graphs, classical rigidity together with inequality constraints, penny/marble graphs and interval radii penny and marble graphs. 

The unit ball graph model only requires that the pairwise distances corresponding to contacts are less than or equal to some constant, and the rest of the pairwise distances are larger than the constant. It is the simplest and least restrictive of the models, and perhaps biologically the one that makes the most sense. However, we proved that the 3D reconstruction is never finitely identifiable under this model. Our reconstruction experiments illustrate this as the 3D reconstructions do not match with the original structures in detail. Nevertheless, the high-level details are correct and the contact matrices are similar for the original and reconstructed structures, so acceptable reconstructions might be still possible under this model. One direction to explore is the geometry and topology of the space of all the valid 3D reconstructions.

The second model that we studied assumes that we know some pairwise distances and some bounds for the pairwise distances between the beads. The analysis of this model is closely related to the classical rigidity setting, where some pairwise distances between a set of points are known. We show that finite and unique identifiability of 3D reconstructions are possible under this model. The finite identifiability of the reconstruction is determined only by the pairwise distances; the unique identifiability depends both on the pairwise distances and bounds on the distances. We obtain good reconstructions also in our experiments, especially when we know enough pairwise distances. While mathematically this model has desirable properties, in practice it might be difficult to obtain the pairwise distance measurements.

Finally, we study penny/marble graph models that are special cases of the previous model where all known pairwise distances are equal to one and all other pairwise distances are bounded below by one. This model is more restrictive than the unit ball model, as all contacts correspond to equal pairwise distances, and therefore this model is biologically less realistic. The motivation for studying this model is the observation that identifiability is possible under this model, although studying it is difficult and not much is known. The synthetic reconstruction experiments gave the worst results under this model. We do not know if the reason is that in the particular cases the reconstructions were not identifiable or whether the assumptions of this model are too restrictive. Based on our analysis, we do not recommend using the penny/marble graph model. 

Our more specialised interval radii penny/marble graph model allows for interesting mathematics. However a limitation of our reconstruction algorithm is that it does not apply to this model. 

As we mentioned in \Cref{sec:preliminaries}, the assumption that the absence of an interaction suggests that the beads are more than $d_c$ units apart in this and other models corresponds to the idealized setting when all contacts are sampled in an Hi-C experiment. Most of our identifiability results remain the same without this assumption. The unit ball model is never finitely identifiable by~\Cref{lemma:uniqueness_threshold_model} and the same remains true when the assumption about non-contacts is dropped. The main uniqueness result for the model with some distance equalities and inequalities~\Cref{prop:finite} does not specify the exact constraints. A consequence of this result is that if we have unique identifiability with the non-contact constraints, then we might get finite identifiability when dropping these constraints. The effect of an imperfect Hi-C matrix for the penny/marble graph model and the $\varepsilon$-interval radii model is to replace every strict inequality for a non-edge with a non-strict inequality. Again, this can only increase the number of solutions, and hence identifiability could be replaced by finitely or infinitely many solutions. In the penny/marble graph model it is also possible for a graph with no solutions to gain solutions, or for an identifiable system to gain infinitely-many solutions. Interestingly, neither of these latter possibilities can occur for the $\varepsilon$-interval radii penny graph model. This is due to the genericity of the radii in conjunction with results of \citet{stickydisks}.

This paper focused on single-cell 3D genome reconstruction in the haploid setting. A future research direction is to conduct a similar study in the diploid setting. We expect the diploid setting, especially the problem of obtaining identifiability results under reasonable constraints, to be considerably more challenging. 

\section*{Acknowledgments}
We thank Caroline Uhler for suggesting to study single cell 3D genome reconstruction and Annachiara Korchmaros for helping to improve the paragraph on Hi-C experiments.

\addcontentsline{toc}{section}{Declarations}
\section*{Declarations}
S.\,D.\ was supported by the Heilbronn Institute for Mathematical Research.
G.\,G.\ was partially supported by the Austrian Science Fund (FWF): 10.55776/P31888.
K.\,K.\ was partially supported by the Academy of Finland grant number 323416.
F.\,M.\ was partially supported by the FWO grants G0F5921N and G023721N, 
the KU Leuven iBOF/23/064 grant, and the UiT Aurora MASCOT project.
A.\,N.\ was partially supported by EPSRC grant EP/X036723/1.

The code used to run the experiments in~\Cref{sec:experiments} is available on GitHub:\\
\url{https://github.com/kaiekubjas/haploid-single-cell-3D-genome-reconstruction}

The real data analysed in~\Cref{sec:experiments} is associated to the article by~\citet{stevens20173d} and it is accessible at NCBI GEO database of \citet{edgar2002gene}, accession GSE80280.


\begin{thebibliography}{68}

\bibitem[Abbas et~al.(2019)Abbas, He, Niu, Zhou, Zhu, Ma, Song, Gao, Zhang, and
  Zeng]{abbas2019integrating}
Ahmed Abbas, Xuan He, Jing Niu, Bin Zhou, Guangxiang Zhu, Tszshan Ma,
  Jiangpeikun Song, Juntao Gao, Michael~Q Zhang, and Jianyang Zeng.
\newblock Integrating {H}i-{C} and {FISH} data for modeling of the 3{D}
  organization of chromosomes.
\newblock \emph{Nature Communications}, 10\penalty0 (1):\penalty0 1--14,
\newblock \href{https://doi.org/10.1038/s41467-019-10005-6}{\path{doi:10.1038/s41467-019-10005-6}}.

\bibitem[Amann et~al.(1990)Amann, Krumholz, and Stahl]{amann1990fluorescent}
Rudolf~I. Amann, Lee Krumholz, and David~A. Stahl.
\newblock Fluorescent-oligonucleotide probing of whole cells for determinative,
  phylogenetic, and environmental studies in microbiology.
\newblock \emph{Journal of Bacteriology}, 172\penalty0 (2):\penalty0 762--770,
\newblock \href{https://doi.org/10.1128/jb.172.2.762-770.1990}{\path{doi:10.1128/jb.172.2.762-770.1990}}.

\bibitem[Asimow and Roth(1978)]{AsimowRoth}
Leonard Asimow and Ben Roth.
\newblock The rigidity of graphs.
\newblock \emph{Transactions of the American Mathematical Society},
  245:\penalty0 279--289,
\newblock \href{https://doi.org/10.2307/1998867}{\path{doi:10.2307/1998867}}.

\bibitem[Atminas and Zamaraev(2018)]{atminas2018forbidden}
Aistis Atminas and Viktor Zamaraev.
\newblock On forbidden induced subgraphs for unit disk graphs.
\newblock \emph{Discrete \& Computational Geometry}, 60\penalty0 (1):\penalty0
  57--97,
\newblock \href{https://doi.org/10.1007/s00454-018-9968-1}{\path{doi:10.1007/s00454-018-9968-1}}.

\bibitem[Belyaeva et~al.(2022)Belyaeva, Kubjas, Sun, and
  Uhler]{belyaeva2021identifying}
Anastasiya Belyaeva, Kaie Kubjas, Lawrence~J Sun, and Caroline Uhler.
\newblock {Identifying 3D Genome Organization in Diploid Organisms via
  Euclidean Distance Geometry}.
\newblock \emph{SIAM Journal on Mathematics of Data Science}, 4\penalty0
  (1):\penalty0 204--228,
\newblock \href{https://doi.org/10.1137/21M1390372}{\path{doi:10.1137/21M1390372}}.

\bibitem[Bezdek(2012)]{Bezdek12}
K\'{a}roly Bezdek.
\newblock Contact numbers for congruent sphere packings in euclidean 3-space.
\newblock \emph{Discrete and Computational Geometry}, 48\penalty0 (2):\penalty0
  298--309,
\newblock \href{https://doi.org/10.1007/s00454-012-9405-9}{\path{doi:10.1007/s00454-012-9405-9}}.

\bibitem[Bezdek and Reid(2013)]{bezdekreid}
Károly Bezdek and Samuel Reid.
\newblock Contact graphs of unit sphere packings revisited.
\newblock \emph{Journal of Geometry}, 104:\penalty0 57--83,
\newblock \href{https://doi.org/10.1007/s00022-013-0156-4}{\path{doi:10.1007/s00022-013-0156-4}}.

\bibitem[Biswas et~al.(2006)Biswas, Liang, Toh, Ye, and
  Wang]{biswas2006semidefinite2}
Pratik Biswas, Tzu-Chen Liang, Kim-Chuan Toh, Yinyu Ye, and Ta-Chung Wang.
\newblock Semidefinite programming approaches for sensor network localization
  with noisy distance measurements.
\newblock \emph{IEEE Transactions on Automation Science and Engineering},
  3\penalty0 (4):\penalty0 360--371,
\newblock \href{https://doi.org/10.1109/TASE.2006.877401}{\path{doi:10.1109/TASE.2006.877401}}.

\bibitem[Breu and Kirkpatrick(1996)]{Breu96}
Heinz Breu and David~G. Kirkpatrick.
\newblock On the complexity of recognizing intersection and touching graphs of
  disks.
\newblock In Franz~J. Brandenburg, editor, \emph{Graph Drawing}, pages 88--98,
  Berlin, Heidelberg, 1996.
\newblock \href{https://doi.org/10.1007/BFb0021793}{\path{doi:10.1007/BFb0021793}}.

\bibitem[Breu and Kirkpatrick(1998)]{Breu98}
Heinz Breu and David~G. Kirkpatrick.
\newblock {Unit disk graph recognition is NP-hard}.
\newblock \emph{Computational Geometry}, 9\penalty0 (1):\penalty0 3--24,
\newblock \href{https://doi.org/10.1016/S0925-7721(97)00014-X}{\path{doi:10.1016/S0925-7721(97)00014-X}}.

\bibitem[Cassioli et~al.(2015)Cassioli, Bardiaux, Bouvier, Mucherino, Alves,
  Liberti, Nilges, Lavor, and Malliavin]{cassioli2015algorithm}
Andrea Cassioli, Benjamin Bardiaux, Guillaume Bouvier, Antonio Mucherino,
  Rafael Alves, Leo Liberti, Michael Nilges, Carlile Lavor, and Th{\'e}rese~E
  Malliavin.
\newblock An algorithm to enumerate all possible protein conformations
  verifying a set of distance constraints.
\newblock \emph{BMC bioinformatics}, 16:\penalty0 1--15,
\newblock \href{https://doi.org/10.1186/s12859-015-0451-1}{\path{doi:10.1186/s12859-015-0451-1}}.

\bibitem[Chen(2016)]{Chen16}
Hao Chen.
\newblock Apollonian ball packings and stacked polytopes.
\newblock \emph{Discrete \& Computational Geometry}, 55:\penalty0 801--826,
\newblock \href{https://doi.org/10.1007/s00454-016-9777-3}{\path{doi:10.1007/s00454-016-9777-3}}.

\bibitem[Cifuentes et~al.(2024)Cifuentes, Draisma, Henriksson, Korchmaros, and
  Kubjas]{cifuentes20233d}
Diego Cifuentes, Jan Draisma, Oskar Henriksson, Annachiara Korchmaros, and Kaie
  Kubjas.
\newblock {3D genome reconstruction from partially phased Hi-C data}.
\newblock \emph{{Bulletin of Mathematical Biology}}, 86\penalty0 (33):\penalty0
  1--30,
\newblock \href{https://doi.org/10.1007/s11538-024-01263-7}{\path{doi:10.1007/s11538-024-01263-7}}.

\bibitem[Connelly(2005)]{Con05}
Robert Connelly.
\newblock Generic global rigidity.
\newblock \emph{Discrete \& Computational Geometry}, 33\penalty0 (4):\penalty0
  549--563,
\newblock \href{https://doi.org/10.1007/s00454-004-1124-4}{\path{doi:10.1007/s00454-004-1124-4}}.

\bibitem[Connelly et~al.(2019)Connelly, Gortler, and Theran]{stickydisks}
Robert Connelly, Steven~J Gortler, and Louis Theran.
\newblock Rigidity for sticky discs.
\newblock \emph{Proceedings of the Royal Society A}, 475\penalty0 (2222),
\newblock \href{https://doi.org/10.1098/rspa.2018.0773}{\path{doi:10.1098/rspa.2018.0773}}.

\bibitem[Cox and Cox(2008)]{cox2008multidimensional}
Michael~{A.\,A.} Cox and Trevor~F. Cox.
\newblock Multidimensional scaling.
\newblock In \emph{Handbook of data visualization}, pages 315--347. Springer,
  Berlin, Heidelberg,
\newblock \href{https://doi.org/10.1007/978-3-540-33037-0_14}{\path{doi:10.1007/978-3-540-33037-0_14}}.

\bibitem[Dekker(2008)]{dekker2008gene}
Job Dekker.
\newblock Gene regulation in the third dimension.
\newblock \emph{Science}, 319\penalty0 (5871):\penalty0 1793--1794,
\newblock \href{https://doi.org/10.1126/science.1152850}{\path{doi:10.1126/science.1152850}}.

\bibitem[Dewar(2023)]{DewarSphere}
Sean Dewar.
\newblock {Identifying contact graphs of sphere packings with generic radii},
  2023.

\bibitem[Dewar et~al.(2023)Dewar, Grasegger, Kubjas, Mohammadi, and
  Nixon]{penny_paper}
Sean Dewar, Georg Grasegger, Kaie Kubjas, Fatemeh Mohammadi, and Anthony Nixon.
\newblock On the uniqueness of collections of pennies and marbles, 2023.

\bibitem[Edgar et~al.(2002)Edgar, Domrachev, and Lash]{edgar2002gene}
Ron Edgar, Michael Domrachev, and Alex~E. Lash.
\newblock {Gene Expression Omnibus: NCBI gene expression and hybridization
  array data repository}.
\newblock \emph{Nucleic acids research}, 30\penalty0 (1):\penalty0 207--210,
\newblock \href{https://doi.org/10.1093/nar/30.1.207}{\path{doi:10.1093/nar/30.1.207}}.

\bibitem[Fazel et~al.(2001)Fazel, Hindi, and Boyd]{fazel2001rank}
Maryam Fazel, Haitham Hindi, and Stephen~P. Boyd.
\newblock A rank minimization heuristic with application to minimum order
  system approximation.
\newblock In \emph{Proceedings of the 2001 American Control Conference},
  volume~6, pages 4734--4739. IEEE,
\newblock \href{https://doi.org/10.1109/ACC.2001.945730}{\path{doi:10.1109/ACC.2001.945730}}.

\bibitem[Garamvolgyi and Jord\'{a}n(2020)]{garamvolgyi2020global}
Daniel Garamvolgyi and Tibor Jord\'{a}n.
\newblock Global rigidity of unit ball graphs.
\newblock \emph{SIAM Journal on Discrete Mathematics}, 34\penalty0
  (1):\penalty0 212--229,
\newblock \href{https://doi.org/10.1137/18M1220662}{\path{doi:10.1137/18M1220662}}.

\bibitem[Gortler et~al.(2010)Gortler, Healy, and
  Thurston]{GortlerHealyThurston}
Steven~J. Gortler, Alexander~D. Healy, and Dylan~P. Thurston.
\newblock Characterizing generic global rigidity.
\newblock \emph{American Journal of Mathematics}, 132\penalty0 (4):\penalty0
  897--939,
\newblock \href{https://doi.org/10.1353/ajm.0.0132}{\path{doi:10.1353/ajm.0.0132}}.

\bibitem[Harborth(1974)]{Harboth}
Heiko Harborth.
\newblock Lösung zu problem 664a.
\newblock \emph{Elemente der Mathematik}, 29\penalty0 (1):\penalty0 14--15,
  1974.

\bibitem[Hendrickson(1992)]{Hen}
Bruce Hendrickson.
\newblock Conditions for unique graph realizations.
\newblock \emph{SIAM Journal on Computing}, 21\penalty0 (1):\penalty0 65--84,
\newblock \href{https://doi.org/10.1137/0221008}{\path{doi:10.1137/0221008}}.

\bibitem[Hlin\v{e}n\'{y}(1997)]{Hlin97}
Petr Hlin\v{e}n\'{y}.
\newblock Touching graphs of unit balls.
\newblock In Giuseppe DiBattista, editor, \emph{Graph Drawing}, pages 350--358,
  Berlin, Heidelberg, 1997.
\newblock \href{https://doi.org/10.1007/3-540-63938-1_80}{\path{doi:10.1007/3-540-63938-1_80}}.

\bibitem[Hlin\v{e}n\'{y} and Kratochv\'{\i}l(2001)]{HK01}
Petr Hlin\v{e}n\'{y} and Jan Kratochv\'{\i}l.
\newblock Representing graphs by disks and balls (a survey of
  recognition-complexity results).
\newblock \emph{Discrete Mathematics}, 229\penalty0 (1):\penalty0 101--124,
\newblock \href{https://doi.org/10.1016/S0012-365X(00)00204-1}{\path{doi:10.1016/S0012-365X(00)00204-1}}.

\bibitem[Holmes-Cerfon(2016)]{HolmesCerfon}
Miranda~C. Holmes-Cerfon.
\newblock Enumerating rigid sphere packings.
\newblock \emph{SIAM Review}, 58\penalty0 (2):\penalty0 229--244,
\newblock \href{https://doi.org/10.1137/140982337}{\path{doi:10.1137/140982337}}.

\bibitem[Hu et~al.(2013)Hu, Deng, Qin, Dixon, Selvaraj, Fang, Ren, and
  Liu]{hu2013bayesian}
Ming Hu, Ke~Deng, Zhaohui Qin, Jesse Dixon, Siddarth Selvaraj, Jennifer Fang,
  Bing Ren, and Jun~S. Liu.
\newblock Bayesian inference of spatial organizations of chromosomes.
\newblock \emph{PLoS Computational Biology}, 9\penalty0 (1):\penalty0 e1002893,
\newblock \href{https://doi.org/10.1371/journal.pcbi.1002893}{\path{doi:10.1371/journal.pcbi.1002893}}.

\bibitem[Jackson and Jord\'{a}n(2005)]{JJddim}
Bill Jackson and Tibor Jord\'{a}n.
\newblock The d-dimensional rigidity matroid of sparse graphs.
\newblock \emph{Journal of Combinatorial Theory, Series B}, 95\penalty0
  (1):\penalty0 118--133,
\newblock \href{https://doi.org/10.1016/j.jctb.2005.03.004}{\path{doi:10.1016/j.jctb.2005.03.004}}.

\bibitem[Jackson and Jordán(2005)]{JacksonJordan}
Bill Jackson and Tibor Jordán.
\newblock Connected rigidity matroids and unique realizations of graphs.
\newblock \emph{Journal of Combinatorial Theory, Series B}, 94\penalty0
  (1):\penalty0 1--29,
\newblock \href{https://doi.org/10.1016/j.jctb.2004.11.002}{\path{doi:10.1016/j.jctb.2004.11.002}}.

\bibitem[Jackson et~al.(2014)Jackson, McCourt, and Nixon]{JMN14}
Bill Jackson, Thomas~A. McCourt, and Anthony Nixon.
\newblock Necessary conditions for the generic global rigidity of frameworks on
  surfaces.
\newblock \emph{Discrete \& Computational Geometry}, 52\penalty0 (2):\penalty0
  344--360,
\newblock \href{https://doi.org/10.1007/s00454-014-9616-3}{\path{doi:10.1007/s00454-014-9616-3}}.

\bibitem[Jacobs and Hendrickson(1997)]{JHpebble}
Donald~J. Jacobs and Bruce Hendrickson.
\newblock An algorithm for two-dimensional rigidity percolation: The pebble
  game.
\newblock \emph{Journal of Computational Physics}, 137\penalty0 (2):\penalty0
  346--365,
\newblock \href{https://doi.org/10.1006/jcph.1997.5809}{\path{doi:10.1006/jcph.1997.5809}}.

\bibitem[Kert\'{e}z(1994)]{Kertez}
G\'{a}bor Kert\'{e}z.
\newblock Nine points on the hemisphere.
\newblock In K.~B\"{o}r\"{o}czky and G.~Fejes~T\'{o}th, editors,
  \emph{Intuitive Geometry (Szeged, 1991)}, volume~63 of \emph{Colloquia
  Mathematica Societatis J\'{a}nos Bolyai}, pages 189--196, Amsterdam, 1994.
  North Holland Publishing co.

\bibitem[Kos et~al.(2021)Kos, Galitsyna, Ulianov, Gelfand, Razin, and
  Chertovich]{kos2021perspectives}
Pavel~I. Kos, Aleksandra~A. Galitsyna, Sergey~V. Ulianov, Mikhail~S. Gelfand,
  Sergey~V. Razin, and Alexander~V. Chertovich.
\newblock Perspectives for the reconstruction of {3D} chromatin conformation
  using single cell {Hi-C} data.
\newblock \emph{PLoS Computational Biology}, 17\penalty0 (11):\penalty0
  e1009546,
\newblock \href{https://doi.org/10.1371/journal.pcbi.1009546}{\path{doi:10.1371/journal.pcbi.1009546}}.

\bibitem[Lekkerkerker and Boland(1962)]{Lekkeikerker1962}
Cornelis~G. Lekkerkerker and Johan~C. Boland.
\newblock Representation of a finite graph by a set of intervals on the real
  line.
\newblock \emph{Fundamenta Mathematicae}, 51\penalty0 (1):\penalty0 45--64,
\newblock \href{https://doi.org/10.4064/fm-51-1-45-64}{\path{doi:10.4064/fm-51-1-45-64}}.

\bibitem[Lesne et~al.(2014)Lesne, Riposo, Roger, Cournac, and
  Mozziconacci]{lesne20143d}
Annick Lesne, Julien Riposo, Paul Roger, Axel Cournac, and Julien Mozziconacci.
\newblock 3{D} genome reconstruction from chromosomal contacts.
\newblock \emph{Nature methods}, 11\penalty0 (11):\penalty0 1141--1143,
\newblock \href{https://doi.org/10.1038/nmeth.3104}{\path{doi:10.1038/nmeth.3104}}.

\bibitem[Leung and Toh(2010)]{leung2010sdp}
Ngai-Hang~Z. Leung and Kim-Chuan Toh.
\newblock An {SDP}-based divide-and-conquer algorithm for large-scale noisy
  anchor-free graph realization.
\newblock \emph{SIAM Journal on Scientific Computing}, 31\penalty0
  (6):\penalty0 4351--4372,
\newblock \href{https://doi.org/10.1137/080733103}{\path{doi:10.1137/080733103}}.

\bibitem[Liberti et~al.(2014)Liberti, Lavor, Maculan, and
  Mucherino]{liberti2014euclidean}
Leo Liberti, Carlile Lavor, Nelson Maculan, and Antonio Mucherino.
\newblock Euclidean distance geometry and applications.
\newblock \emph{SIAM Review}, 56\penalty0 (1):\penalty0 3--69,
\newblock \href{https://doi.org/10.1137/120875909}{\path{doi:10.1137/120875909}}.

\bibitem[Lieberman-Aiden et~al.(2009)Lieberman-Aiden, Van~Berkum, Williams,
  Imakaev, Ragoczy, Telling, Amit, Lajoie, Sabo, Dorschner,
  et~al.]{lieberman2009comprehensive}
Erez Lieberman-Aiden, Nynke~L. Van~Berkum, Louise Williams, Maxim Imakaev,
  Tobias Ragoczy, Agnes Telling, Ido Amit, Bryan~R. Lajoie, Peter~J. Sabo,
  Michael~O. Dorschner, et~al.
\newblock Comprehensive mapping of long-range interactions reveals folding
  principles of the human genome.
\newblock \emph{Science}, 326\penalty0 (5950):\penalty0 289--293,
\newblock \href{https://doi.org/10.1126/science.1181369}{\path{doi:10.1126/science.1181369}}.

\bibitem[Maehara(2007)]{Maehara}
Hiroshi Maehara.
\newblock On configurations of solid balls in 3-space: chromatic numbers and
  knotted cycles.
\newblock \emph{Graphs and Combinatorics}, 23:\penalty0 307--320,
\newblock \href{https://doi.org/10.1007/s00373-007-0702-7}{\path{doi:10.1007/s00373-007-0702-7}}.

\bibitem[Maxwell(1864)]{Maxwell}
James~Clerk Maxwell.
\newblock On the calculation of the equilibrium and stiffness of frames.
\newblock \emph{The London, Edinburgh, and Dublin Philosophical Magazine and
  Journal of Science}, 27\penalty0 (182):\penalty0 294--299,
\newblock \href{https://doi.org/10.1080/14786446408643668}{\path{doi:10.1080/14786446408643668}}.

\bibitem[Milnor(1965)]{Milnor1965}
John Milnor.
\newblock \emph{Topology from the Differentiable Viewpoint}.
\newblock University Press of Virginia, Charlottesville, 1965.

\bibitem[Nagano et~al.(2013)Nagano, Lubling, Stevens, Schoenfelder, Yaffe,
  Dean, Laue, Tanay, and Fraser]{nagano2013single}
Takashi Nagano, Yaniv Lubling, Tim~J. Stevens, Stefan Schoenfelder, Eitan
  Yaffe, Wendy Dean, Ernest~D. Laue, Amos Tanay, and Peter Fraser.
\newblock Single-cell {Hi-C} reveals cell-to-cell variability in chromosome
  structure.
\newblock \emph{Nature}, 502\penalty0 (7469):\penalty0 59--64,
\newblock \href{https://doi.org/10.1038/nature12593}{\path{doi:10.1038/nature12593}}.

\bibitem[Norton and Phillips-Cremins(2017)]{norton2017crossed}
Heidi~K. Norton and Jennifer~E. Phillips-Cremins.
\newblock Crossed wires: 3{D} genome misfolding in human disease.
\newblock \emph{Journal of Cell Biology}, 216\penalty0 (11):\penalty0
  3441--3452,
\newblock \href{https://doi.org/10.1083/jcb.201611001}{\path{doi:10.1083/jcb.201611001}}.

\bibitem[Oluwadare et~al.(2019)Oluwadare, Highsmith, and
  Cheng]{oluwadare2019overview}
Oluwatosin Oluwadare, Max Highsmith, and Jianlin Cheng.
\newblock An overview of methods for reconstructing 3-{D} chromosome and genome
  structures from {Hi-C} data.
\newblock \emph{Biological Procedures Online}, 21\penalty0 (1):\penalty0 1--20,
\newblock \href{https://doi.org/10.1186/s12575-019-0094-0}{\path{doi:10.1186/s12575-019-0094-0}}.

\bibitem[Ozkan et~al.(2018)Ozkan, Prabhu, Baker, Pence, Peters, and
  Sitharam]{ozkan2018algorithm}
Aysegul Ozkan, Rahul Prabhu, Troy Baker, James Pence, Jorg Peters, and Meera
  Sitharam.
\newblock Algorithm 990: efficient atlasing and search of configuration spaces
  of point-sets constrained by distance intervals.
\newblock \emph{ACM Transactions on Mathematical Software (TOMS)}, 44\penalty0
  (4):\penalty0 1--30,
\newblock \href{https://doi.org/10.1145/3204472}{\path{doi:10.1145/3204472}}.

\bibitem[Paulsen et~al.(2015)Paulsen, Gramstad, and
  Collas]{paulsen2015manifold}
Jonas Paulsen, Odin Gramstad, and Philippe Collas.
\newblock Manifold based optimization for single-cell {3D} genome
  reconstruction.
\newblock \emph{PLoS computational biology}, 11\penalty0 (8):\penalty0
  e1004396,
\newblock \href{https://doi.org/10.1371/journal.pcbi.1004396}{\path{doi:10.1371/journal.pcbi.1004396}}.

\bibitem[Paulsen et~al.(2017)Paulsen, Sekelja, Oldenburg, Barateau, Briand,
  Delbarre, Shah, S{\o}rensen, Vigouroux, Buendia, et~al.]{paulsen2017chrom3d}
Jonas Paulsen, Monika Sekelja, Anja~R. Oldenburg, Alice Barateau, Nolwenn
  Briand, Erwan Delbarre, Akshay Shah, Anita~L. S{\o}rensen, Corinne Vigouroux,
  Brigitte Buendia, et~al.
\newblock Chrom3{D}: three-dimensional genome modeling from {Hi-C} and nuclear
  lamin-genome contacts.
\newblock \emph{Genome Biology}, 18\penalty0 (1):\penalty0 1--15,
\newblock \href{https://doi.org/10.1186/s13059-016-1146-2}{\path{doi:10.1186/s13059-016-1146-2}}.

\bibitem[Payne et~al.(2021)Payne, Chiang, Reginato, Mangiameli, Murray, Yao,
  Markoulaki, Earl, Labade, Jaenisch, et~al.]{payne2021situ}
Andrew~C. Payne, Zachary~D. Chiang, Paul~L. Reginato, Sarah~M. Mangiameli,
  Evan~M. Murray, Chun-Chen Yao, Styliani Markoulaki, Andrew~S. Earl, Ajay~S.
  Labade, Rudolf Jaenisch, et~al.
\newblock In situ genome sequencing resolves {DNA} sequence and structure in
  intact biological samples.
\newblock \emph{Science}, 371\penalty0 (6532):\penalty0 eaay3446,
\newblock \href{https://doi.org/10.1126/science.aay3446}{\path{doi:10.1126/science.aay3446}}.

\bibitem[Pfender and Ziegler(2004)]{PfenderZiegler}
Florian Pfender and Günter~M. Ziegler.
\newblock Kissing numbers, sphere packings, and some unexpected proofs.
\newblock \emph{Notices of the AMS}, 51\penalty0 (8):\penalty0 873--883, 2004.

\bibitem[Pollaczek-Geiringer(1927)]{Geiringer}
Hilda Pollaczek-Geiringer.
\newblock {Über die Gliederung ebener Fachwerke}.
\newblock \emph{{Zeitschrift für Angewandte Mathematik und Mechanik}},
  7\penalty0 (1):\penalty0 58--72,
\newblock \href{https://doi.org/10.1002/zamm.19270070107}{\path{doi:10.1002/zamm.19270070107}}.

\bibitem[Ramani et~al.(2017)Ramani, Deng, Qiu, Gunderson, Steemers, Disteche,
  Noble, Duan, and Shendure]{ramani2017massively}
Vijay Ramani, Xinxian Deng, Ruolan Qiu, Kevin~L. Gunderson, Frank~J. Steemers,
  Christine~M. Disteche, William~S. Noble, Zhijun Duan, and Jay Shendure.
\newblock Massively multiplex single-cell {Hi-C}.
\newblock \emph{Nature Methods}, 14\penalty0 (3):\penalty0 263--266,
\newblock \href{https://doi.org/10.1038/nmeth.4155}{\path{doi:10.1038/nmeth.4155}}.

\bibitem[Roberts(1969)]{indifference}
Fred~S. Roberts.
\newblock Indifference graphs.
\newblock In \emph{{Proof Techniques in Graph Theory (Proceedings of the Second
  Ann Arbor Graph Theory Conference, 1968)}}, pages 139--146, New York, 1969.
  Academic Press.

\bibitem[Rosenthal et~al.(2019)Rosenthal, Bryner, Huffer, Evans, Srivastava,
  and Neretti]{rosenthal2019bayesian}
Michael Rosenthal, Darshan Bryner, Fred Huffer, Shane Evans, Anuj Srivastava,
  and Nicola Neretti.
\newblock Bayesian estimation of three-dimensional chromosomal structure from
  single-cell {Hi-C} data.
\newblock \emph{Journal of Computational Biology}, 26\penalty0 (11):\penalty0
  1191--1202,
\newblock \href{https://doi.org/10.1089/cmb.2019.0100}{\path{doi:10.1089/cmb.2019.0100}}.

\bibitem[Schoenberg(1935)]{schoenberg1935remarks}
Isaac~J. Schoenberg.
\newblock Remarks to {M}aurice {F}r\'{e}chet's article ``sur la definition
  axiomatique d'une classe d'espace distances vectoriellement applicable sur
  l'espace de {H}ilbert''.
\newblock \emph{Annals of Mathematics}, pages 724--732,
\newblock \href{https://doi.org/10.2307/1968654}{\path{doi:10.2307/1968654}}.

\bibitem[Schütte and van~der Waerden(1952)]{Kissing3}
Kurt Schütte and Bartel~L. van~der Waerden.
\newblock {Das Problem der dreizehn Kugeln}.
\newblock \emph{Mathematische Annalen}, 125:\penalty0 325--334,
\newblock \href{https://doi.org/10.1007/BF01343127}{\path{doi:10.1007/BF01343127}}.

\bibitem[Segal(2022)]{segal2022can}
Mark~R. Segal.
\newblock Can 3{D} diploid genome reconstruction from unphased {H}i-{C} data be
  salvaged?
\newblock \emph{NAR Genomics and Bioinformatics}, 4\penalty0 (2):\penalty0
  lqac038,
\newblock \href{https://doi.org/10.1093/nargab/lqac038}{\path{doi:10.1093/nargab/lqac038}}.

\bibitem[Shi and Thirumalai(2021)]{shi2021hi}
Guang Shi and Dave Thirumalai.
\newblock From {Hi-C} contact map to three-dimensional organization of
  interphase human chromosomes.
\newblock \emph{Physical Review X}, 11\penalty0 (1):\penalty0 011051,
\newblock \href{https://doi.org/10.1103/PhysRevX.11.011051}{\path{doi:10.1103/PhysRevX.11.011051}}.

\bibitem[So and Ye(2007)]{so2007theory}
Anthony Man-Cho So and Yinyu Ye.
\newblock Theory of semidefinite programming for sensor network localization.
\newblock \emph{Mathematical Programming}, 109\penalty0 (2):\penalty0 367--384,
\newblock \href{https://doi.org/10.1007/s10107-006-0040-1}{\path{doi:10.1007/s10107-006-0040-1}}.

\bibitem[Stevens et~al.(2017)Stevens, Lando, Basu, Atkinson, Cao, Lee, Leeb,
  Wohlfahrt, Boucher, O’Shaughnessy-Kirwan, et~al.]{stevens20173d}
Tim~J. Stevens, David Lando, Srinjan Basu, Liam~P. Atkinson, Yang Cao,
  Steven~F. Lee, Martin Leeb, Kai~J. Wohlfahrt, Wayne Boucher, Aoife
  O’Shaughnessy-Kirwan, et~al.
\newblock {3D structures of individual mammalian genomes studied by single-cell
  Hi-C}.
\newblock \emph{Nature}, 544\penalty0 (7648):\penalty0 59--64,
\newblock \href{https://doi.org/10.1038/nature21429}{\path{doi:10.1038/nature21429}}.

\bibitem[Tan et~al.(2018)Tan, Xing, Chang, Li, and Xie]{tan2018three}
Longzhi Tan, Dong Xing, Chi-Han Chang, Heng Li, and X~Sunney Xie.
\newblock Three-dimensional genome structures of single diploid human cells.
\newblock \emph{Science}, 361\penalty0 (6405):\penalty0 924--928,
\newblock \href{https://doi.org/10.1126/science.aat5641}{\path{doi:10.1126/science.aat5641}}.

\bibitem[Trieu and Cheng(2014)]{trieu2014large}
Tuan Trieu and Jianlin Cheng.
\newblock Large-scale reconstruction of 3{D} structures of human chromosomes
  from chromosomal contact data.
\newblock \emph{Nucleic Acids Research}, 42\penalty0 (7):\penalty0 e52--e52,
\newblock \href{https://doi.org/10.1093/nar/gkt1411}{\path{doi:10.1093/nar/gkt1411}}.

\bibitem[Uhler and Shivashankar(2017)]{uhler2017regulation}
Caroline Uhler and G.V. Shivashankar.
\newblock Regulation of genome organization and gene expression by nuclear
  mechanotransduction.
\newblock \emph{Nature Reviews Molecular Cell Biology}, 18\penalty0
  (12):\penalty0 717--727,
\newblock \href{https://doi.org/10.1038/nrm.2017.101}{\path{doi:10.1038/nrm.2017.101}}.

\bibitem[Varoquaux et~al.(2014)Varoquaux, Ay, Noble, and
  Vert]{varoquaux2014statistical}
Nelle Varoquaux, Ferhat Ay, William~Stafford Noble, and Jean-Philippe Vert.
\newblock A statistical approach for inferring the 3{D} structure of the
  genome.
\newblock \emph{Bioinformatics}, 30\penalty0 (12):\penalty0 i26--i33,
\newblock \href{https://doi.org/10.1093/bioinformatics/btu268}{\path{doi:10.1093/bioinformatics/btu268}}.

\bibitem[Weinberger and Saul(2006)]{weinberger2006unsupervised}
Kilian~Q. Weinberger and Lawrence~K. Saul.
\newblock Unsupervised learning of image manifolds by semidefinite programming.
\newblock \emph{International Journal of Computer Vision}, 70\penalty0
  (1):\penalty0 77--90,
\newblock \href{https://doi.org/10.1007/s11263-005-4939-z}{\path{doi:10.1007/s11263-005-4939-z}}.

\bibitem[Wettermann et~al.(2020)Wettermann, Brems, Siebert, Vu, Stevens, and
  Virnau]{wettermann2020minimal}
Sarah Wettermann, Maarten Brems, Jonathan~T. Siebert, Giang~T. Vu, Tim~J.
  Stevens, and Peter Virnau.
\newblock A minimal {G}{\ = o}-model for rebuilding whole genome structures
  from haploid single-cell {Hi-C} data.
\newblock \emph{Computational Materials Science}, 173:\penalty0 109178,
\newblock \href{https://doi.org/10.1016/j.commatsci.2019.109178}{\path{doi:10.1016/j.commatsci.2019.109178}}.

\bibitem[Zhang et~al.(2013)Zhang, Li, Toh, and Sung]{zhang2013inference}
ZhiZhuo Zhang, Guoliang Li, Kim-Chuan Toh, and Wing-Kin Sung.
\newblock Inference of spatial organizations of chromosomes using semi-definite
  embedding approach and {H}i-{C} data.
\newblock In M.~Deng, R.~Jiang, F.~Sun, and X.~Zhang, editors, \emph{Research
  in Computational Molecular Biology. RECOMB}, pages 317--332. Springer,
\newblock \href{https://doi.org/10.1007/978-3-642-37195-0_31}{\path{doi:10.1007/978-3-642-37195-0_31}}.

\end{thebibliography}
\end{document}